\documentclass[english,course]{Notes}
\usepackage{tikz,color,graphics,stmaryrd,blkarray,amsmath,esint}
\usetikzlibrary{matrix}
\title{Topological recursion and geometry}
\author{Ga\"etan Borot}
\email{Max-Planck Institut for Mathematics: gborot@mpim-bonn.mpg.de}
\website{Website: guests.mpim-bonn.mpg.de/gborot/Teach.html}
\date{12}{05}{2017}
\place{Mini-course, Institut Math\'ematique de Toulouse, France.}

\renewcommand{\thesection}{\arabic{section}}

\definecolor{rouge}{rgb}{0.84,0.18,0.07}
\definecolor{bleu}{rgb}{0.22,0.41,0.74}
\definecolor{vertf}{rgb}{0.08,0.46,0.07}
\usepackage[pdftex]{hyperref}
\hypersetup{colorlinks,urlcolor=vertf,citecolor=rouge,linkcolor=bleu,filecolor=black}

\begin{document}

\newpage

These are lecture notes for a 4h mini-course held in Toulouse, May 9-12th, at the thematic school on
\begin{center}
\textbf{Quantum topology and geometry}
\end{center}
I thank Francesco Costantino and Thomas Fiedler for organization of this event, the audience and especially Rinat Kashaev and Gregor Masbaum for questions and remarks which improved the lectures. \\

Imagine you have a problem to solve concerning surfaces of genus $g$ with $n$ boundaries/marked points. \emph{A topological recursion} is a strategy consisting of
\begin{itemize}
\item[$(1)$] solving the problem for the simplest topologies, \textit{i.e.} disks $(g,n) = (0,1)$, cylinders $(g,n) = (0,2)$, pairs of pants $(g,n) = (0,3)$.
\item[$(2)$] understanding how the problem behaves when one glues together pairs of pants -- or conversely, when one removes a pair of pants/degenerates the surface. You can then solve the problem by induction on $\chi_{g,n} = 2g - 2 + n = -$(the Euler characteristic of the surface).
\end{itemize}
This mantra is realized in several different ways in geometry and quantum field theories. \emph{The topological recursion} (\textsc{tr}) in these lectures refers to an abstract formalism tailored to handle this kind of recursion at the numerical level. When \textsc{tr} applies, it often shadows finer geometric properties which depend on the problem at studied.

The goal of these lectures is to (a) explain some incarnations, in the last ten years, of the idea of topological recursion: in two dimensional quantum field theories (\textsc{2d tqft}s), in cohomological field theories (\textsc{cohft}), in the computation of volumes of the moduli space of curves;  (b) relate them to \textsc{tr}.

\vspace{0.2cm}

\noindent \textbf{Sources.}  The perspective we adopt here on \textsc{tr} was proposed lately by Kontsevich and Soibelman \cite{KSTR} based on the notion of quantum Airy structure. This was further studied by Andersen, the author, Chekhov and Orantin \cite{TRABCD}, and we often borrow arguments from this paper. The original formulation of \textsc{tr} was proposed by Eynard and Orantin \cite{EOFg} using spectral curves, but quantum Airy structures allows a simpler, more algebraic, and slightly more general presentation of the possible initial data for \textsc{tr}. Part of Chapter 2 on \textsc{2d tqft}s is inspired from \cite{Abrams}. The exposition of Givental group action in Chapter 4 is inspired from \cite{Zvonkin}. Chapter 5 is inspired from Mirzakhani \cite{Mirza1} and Wolpert exposition of her works \cite{Wolpertlecture}.

\newpage

\lecture[1h]{09}{05}{2017}

\section{Quantum Airy structures}

\subsection{First principles}

The initial data for \textsc{tr} is a ``quantum Airy structure''. This notion was introduced by Kontsevich and Soibelman in \cite{KSTR}. Let $V$ be a vector space over a field $\mathbb{K}$ of characteristic $0$. Let $(e_i)_{i \in I}$ be a basis of $V$, and $(x_i)_{i \in I}$ the corresponding basis of linear coordinates. $\hbar$ denotes a formal parameter. If $V$ is infinite dimensional, we may have to add assumptions (filtration, completion, convergence) so that all expressions we write make sense. We will overlook this issue, and in the examples of quantum Airy structures with $\dim V = \infty$ we meet, it is possible to check that seemingly infinite sums we are going to write actually contain only finitely many non-zero terms.

\begin{definition}
\label{Def1} A \textbf{quantum Airy structure} on $V$ is the data of a family of differential operators $(L_i)_{i \in I}$, of the form
\beq
\label{Liform} L_i = \hbar \partial_{x_i} - \sum_{a,b \in I} \big( \tfrac{1}{2}\,A^i_{a,b}x_{a}x_{b} + \hbar\,B^i_{a,b}x_{a}\partial_{x_b} + \tfrac{\hbar^2}{2}\,C^i_{a,b}\partial_{x_a}\partial_{x_b}\big) - \hbar D^i\,,
\eeq
and such that 
\beq
\label{Lie}\forall i,j \in I,\qquad [L_i,L_j] = \sum_{a \in I} \hbar\,f_{i,j}^a\,L_{a}\,,
\eeq
where $A^i_{j,k} = A^i_{k,j}$, $B^i_{j,k}$, $C^i_{j,k} = C^i_{k,j}$, $D^i$, $f^i_{j,k} = - f^i_{k,j}$ are scalars indexed by $i,j,k \in I$.
\end{definition}

\vspace{0.2cm}

\noindent $\bullet$ \textbf{The six conditions.} The condition \eqref{Lie} says that $(L_i)_{i \in I}$ span a Lie algebra with structure constants $(f_{i,j}^k)_{i,j,k \in I}$. This is equivalent to a system of six constraints on $(A,B,C,D,f)$.

The identification of the $\partial_{x_k}$ terms, and of the $x_k$ terms, in both sides of \eqref{Lie} yields the two linear relations
\beq
\label{fAeq} \forall i,j,k \in I,\qquad f_{i,j}^k = B^i_{j,k} - B^j_{i,k},\qquad A^i_{j,k} = A^j_{i,k}\,.
\eeq
We use these two relations to get rid of $f$ in the four remaining relations. Identification of the constant term yields
\beq
\label{Deq}\forall i,j,k,\ell \in I,\qquad \sum_{a \in I} B^i_{j,a}D^a + \sum_{a,b \in I} C^i_{a,b}A^{j}_{a,b} = (i \leftrightarrow j)\,.
\eeq
If $(A,B,C)$ are known, this relation is affine in $D$. Then, identification of $x_kx_{\ell}$, $x_k\partial_{x_{\ell}}$ and $\partial_{x_k}\partial_{x_{\ell}}$ yields
\beq
\label{ABCeq} \forall i,j,k,\ell \in I,\qquad \left\{\begin{array}{rcl} \sum_{a \in I} B^i_{j,a}A^a_{k,\ell} + B^i_{k,a}A^j_{a,\ell} + B^i_{\ell,a}A^j_{a,k} & = & (i \leftrightarrow j) \\ \sum_{a \in I} B^i_{j,a}B^a_{k,\ell} + B^i_{k,a}B^j_{a,\ell} + C^i_{\ell,a}A^j_{a,k} & = & (i \leftrightarrow j) \\ \sum_{a \in I} B^i_{j,a}C^a_{k,\ell} + C^i_{k,a}B^i_{a,\ell} + C^i_{\ell,a}B^i_{a,\ell} & = & (i \leftrightarrow j) \end{array}\right.\,.
\eeq
These three quadratic relations have the same index structure. In fact, they form a system of three coupled IHX-type relations (Figure~\ref{FigIHX}). The relation involving $D$ is depicted in Figure~\ref{Drel}.

\begin{center}
\begin{figure}[h!]
\begin{center}\includegraphics[width=0.8\textwidth]{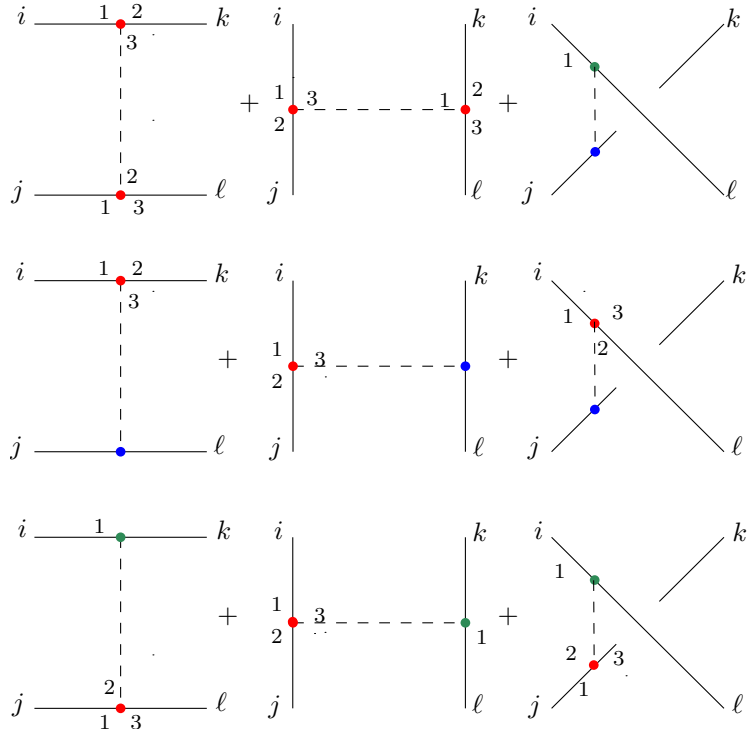}\end{center}
\caption{\label{FigIHX} A red dot is replaced by a $B$, a green dot by a $C$, a blue dot by an $A$. The labels $1,2,3$ specify the ordering of indices at a $B$-vertex, namely $B^{i_{1}}_{i_2,i_3}$, and $1$ at $C$-vertex specifies the first index $C^{i_1}_{**}$. Such an ordering is not necessary at an $A$-vertex as $A$ is fully symmetric. The dashed line carries an index $a \in I$ which is summed over. Equation \eqref{ABCeq} says that each of the three lines is symmetric under permutation of $i$ and $j$.}
\end{figure}

\begin{figure}[h!]
\begin{center}\includegraphics[width=0.6\textwidth]{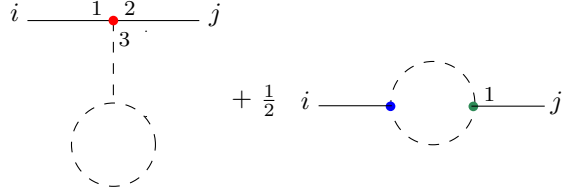}\end{center}
\caption{\label{Drel} The loop carries an index $a \in I$ which is summed over, and stands for $D^a$. Equation \eqref{Deq} says that this expression is symmetric under permutation of $i$ and $j$.}
\end{figure}
\end{center}

\vspace{0.2cm}

\noindent $\bullet$ \textbf{Counting.} If $\dim V = d$, let us count the number of unknowns and independent conditions determining a quantum Airy structure. As \eqref{fAeq} fixes $f$ in terms of $B$, and $A$ is fully symmetric, there remains
$$
\# {\rm unknowns} = \underbrace{\tfrac{d(d + 1)(d + 2)}{6}}_{A} + \underbrace{d^3}_{B} + \underbrace{\tfrac{d^2(d + 1)}{2}}_{C} + \underbrace{d}_{D} = \tfrac{d(5d^2 + 3d + 4)}{3}\,.
$$
The independent relations are indexed by the unordered pair $\{i,j\} \subseteq I$, and for \eqref{ABCeq} by $k,\ell \in I$. In fact, $k$ and $\ell$ play a symmetric role in the BA and BC relations. So
$$
\#{\rm relations} = \tfrac{d(d - 1)}{2}\big(1 +\tfrac{d(d + 1)}{2} + d^2 + \tfrac{d(d + 1)}{2}\big) = \tfrac{d(d - 1)(2d^2 + d + 1)}{2}\,.
$$ 
The number of relations grows faster -- like $d^4$ -- than the number of unknowns -- which grows like $d^3$. For small values of $d$, we have 
$$
\begin{array}{|c|c|c|c|c|c|c|}
\hline
d & 1 & 2 & 3 & 4 & 5 & 6 \\
\hline
\hline \#\,{\rm unknowns} & 4 & 20 & 58 & 128 & 240 & 404 \\
\hline \#\,{\rm relations} & 0 & 11 & 66 & 222 & 560 & 1185 \\
\hline
\end{array}
$$
Therefore, for $d \geq 3$ the system of relations is overdetermined, and it is not clear that quantum Airy structures can exist at all. In fact, the size of this system of quadratic relations makes difficult, even for small $d$, to obtain solutions by brute computer force. We will see that  quantum Airy structure do exist, and in all examples we know of, they are tied to geometry.

\vspace{0.2cm}

\noindent $\bullet$ \textbf{Basis-free definition.} Although we presented the definition of quantum Airy structures using a basis, we now explain how it can be made a basis-independent notion. Let $\mathcal{W}_{V}$ be the Weyl algebra of $V$, that is the quotient of the free algebra generated by $\mathbb{K}[\hbar] \oplus V \oplus V^*$, by the relations
$$
\forall v_1,v_2 \in V,\quad \forall \ell_1,\ell_2 \in V^*,\qquad [\ell_1,\ell_2] = [v_1,v_2] = 0,\qquad [v_1,\ell_1] = \hbar \ell_1(v_1)\,.
$$
The elements of $V \subset \mathcal{W}$ are the derivations $\hbar\partial_{x_i}$, and the elements of $V^* \subseteq \mathcal{W}_{V}$ are the linear coordinates $x_i$ on $V$. We equip $\mathcal{W}_{V}$ with two notions of degree, compatible with the algebra structure. The first one counts $x_i$ and $\hbar\partial_{x_i}$, namely ${\rm deg}\,V = {\rm deg}\,V^* = 1$ and ${\rm deg}\,\mathbb{K}[\hbar] = 0$. We denote $\mathcal{W}_{V}(d)$ -- resp. $\mathcal{W}_{V}(\leq d)$ -- the subspace spanned by elements of degree $d$ -- respectively, less or equal to $d$. The second one is the usual $\hbar$-degree on $\mathbb{K}[\hbar]$, supplemented by the assignments ${\rm deg}_{\hbar} V = {\rm deg}_{\hbar}\,V^* = 0$. In other words, $\hbar^{k}x^{\ell}\partial^{m}$ has $\hbar$-degree $(k - m)$. This convention may seem unnatural, but it is convenient for the next definition.

We remark that $\mathcal{W}_{V}(\leq 2)$ is a Lie algebra, while $\mathcal{W}_{V}(\leq d)$ is not $d \geq 3$ as can be seen from degree counting. This justifies considering operators of degree at most $2$. We denote $\pi_{d}$ and $\pi_{\leq d}$ the projection onto the subspaces $\mathcal{W}_{V}(d)$ and $\mathcal{W}_{V}(\leq d$). 

\begin{definition}
\label{Def2}A quantum Airy structure on $V$ is the data of a Lie algebra structure on $V$ and a homomorphism of Lie algebras $L\,:\,V \rightarrow \mathcal{W}_{V}(\leq 2)$ such that
\begin{itemize}
\item[$\bullet$] ${\rm Im}(\pi_{1} \circ L) = V \subseteq \mathcal{W}_{V}(1)$, and $\pi_{1} \circ L$ induces an isomorphism $\varphi\,:\,V \rightarrow V \subseteq \mathcal{W}_{V}(1)$.
\item[$\bullet$] for $d \in \{0,1,2\}$, we have ${\rm deg}_{\hbar}(\pi_{d}\circ L) = \delta_{d,0}$.
\end{itemize}
\end{definition}

As $\mathcal{W}_{V}(2) = \big({\rm Sym}^2 (V^* \oplus V)\big)[\hbar]$ and taking into account the $\hbar$-degree condition, $\pi_{2}\circ L$ decomposes into three tensors. Turning a $V^*$ in the target to a $V$ in the source, they can be arranged as
$$
A \in {\rm Hom}_{\mathbb{K}}(V^{\otimes 3},\mathbb{K}),\qquad B \in {\rm Hom}_{\mathbb{K}}(V^{\otimes 2},V),\qquad C \in {\rm Hom}_{\mathbb{K}}(V,V^{\otimes 2})\,.
$$
Besides, $\hbar^{-1}(\pi_{0} \circ L)$ gives
$$
D \in {\rm Hom}_{\mathbb{K}}(V,\mathbb{K})\,.
$$
The conditions \eqref{fAeq}-\eqref{Deq}-\eqref{ABCeq} on $(A,B,C,D)$ expressing that $L$ is a Lie algebra morphism can written solely in terms of composition of these morphisms. In fact, they would make sense in the more general context of $V$ being an object in a symmetric monoidal category -- here we worked with the category $\mathbf{Vect}_{\mathbb{K}}$ of finite-dimensional $\mathbb{K}$-vector fields.

The relation between Definitions~\ref{Def1} and \ref{Def2} is as follows. If $(e_i)_{i \in I}$ is a basis of $V$, we should consider it as a basis of $V \subseteq \mathcal{W}_{V}$ and then $L_i = L(\varphi^{-1}(e_i))$ for all $i \in I$.

\subsection{Partition function and topological recursion}
\label{Section1}
Taking a quantum Airy structure $L$ as input,  we can get a function on $V$ which is simultaneously annihilated by the differential operators given by $L$. More precisely, the output $F$ will belong to the space
\beq
\label{Edef}\mathcal{E}_{V} =  \hbar^{-1}\,{\rm Sym}^{> 0} V^*[[\hbar]]\,.
\eeq
We equip $\mathcal{E}_{V}$ with the Euler degree, defined through ${\rm deg}_{\chi}\,\hbar^{g - 1}\,{\rm Sym}^{n} V^* = \chi_{g,n} = 2g - 2 + n$. Elements $F \in \mathcal{E}_{V}$ can be decomposed
$$
F = \sum_{g \geq 0} \sum_{n \geq 1} \frac{\hbar^{g - 1}}{n!}\,F_{g,n},\qquad F_{g,n} \in {\rm Sym}^n(V^*)\,.
$$
We make a small abuse of vocabulary by saying that the Euler degree of $F_{g,n}$ is $\chi_{g,n}$, and we declare that the Euler degree of a product is the sum of the Euler degrees.

\begin{theorem}\cite{KSTR,TRABCD}
\label{MainTh} If $L$ is a quantum Airy structure on $V$, there exists a unique $F \in \mathcal{E}_{V}$ such that $F_{0,1} = 0$, $F_{0,2} = 0$, and for all $v \in V$ we have
$$
L(v)\exp(F) = 0\,.
$$.
\end{theorem}
To fix the vocabulary, we will call $F_{g,n}$ the \textsc{tr} amplitudes, $F$ the ``free energy'', and $Z = \exp(F)$ the ``partition function''.

\vspace{0.2cm}

\noindent $\bullet$ \textbf{Proof of uniqueness.} For convenience, we are going to use a basis of linear coordinates $(x_i)_{i \in I}$ as in Definition~\ref{Def1}. We also agree that indices $a,b,\ldots$ should be summed over the set $I$, while indices $i,j,k,\ldots$ are fixed. If we decompose
$$
F = \sum_{g \geq 0} \hbar^{g - 1}\,F_{g}\,,
$$
we find that the coefficient of $\hbar^{g}$ in $\exp(-F)L_i\exp(F)$ is
\bea
\partial_{x_i} F_{g} & = & \delta_{g,0} \tfrac{1}{2} A^i_{a,b}x_ax_b + B^i_{a,b}x_a\partial_{x_b}F_{g} + \tfrac{1}{2}C^i_{a,b}\Big(\partial_{x_a}\partial_{x_b} F_{g - 1} + \sum_{g_1 + g_2 = g} \partial_{x_a}F_{g_1}\partial_{x_b}F_{g_2}\Big) \nonumber \\
&& + \delta_{g,1}D^i\,. \nonumber
\eea
Let us decompose further
$$
F_g = \sum_{n \geq 1} \sum_{i_1,\ldots,i_n \in I} F_{g,n}[i_1,\ldots,i_n]\,\frac{x_{i_1}\cdots x_{i_n}}{n!}\,.
$$
Now, for $n \geq 1$ we fix an unordered $(n - 1)$-uple of indices $i_2,\ldots,i_n \in I$ and collect the coefficient of the monomial $\tfrac{x_{i_2}\cdots x_{i_n}}{(n - 1)!}$ in this equation. Taking into account the fact that $i_2,\ldots,i_n$ are unordered and $F_{g,n}$ are symmetric, we get
\bea
&& F_{g,n}[i,i_2,\ldots,i_n] \nonumber \\
& = & \delta_{g,0}\delta_{n,3} A^i_{i_2,i_3} + \delta_{g,1}\delta_{n,1}D^i + \sum_{m = 2}^n B^{i}_{i_m,a}\,F_{g,n - 1}[a,i_2,\ldots,\widehat{i_m},\ldots,i_n] \nonumber \\
&& + \tfrac{1}{2} C^i_{a,b}\bigg(F_{g - 1,n + 1}[a,b,i_2,\ldots,i_n] + \sum_{\substack{g_1 + g_2 = g \\ J_1 \dot{\cup} J_2 = \{i_2,\ldots,i_n\}}} F_{g_1,1+|J_1|}[a,J_1] F_{g_2,1 + |J_2|}[b,J_2]\bigg)\,. \nonumber \\
\label{TReq} 
\eea

As $F_{0,1}$ and $F_{0,2}$ are required to vanish, all possibly non-zero $F_{g,n}$ have Euler degree $\chi_{g,n} > 0$. According to \eqref{TReq}, the cases $\chi_{g,n} = 1$ correspond to
\beq
\label{inieq}F_{0,3}[i,i_2,i_3] = A^i_{i_2,i_3},\qquad F_{1,1}[i] = D^i\,,
\eeq
and for $\chi_{g,n} \geq 2$, we have
\bea
&& F_{g,n}[i,i_2,\ldots,i_n] \nonumber \\
& = & \sum_{m = 2}^n B^{i}_{i_m,a}\,F_{g,n - 1}[a,i_2,\ldots,\widehat{i_m},\ldots,i_n] \nonumber \\
&& + \tfrac{1}{2} C^i_{a,b}\bigg(F_{g - 1,n + 1}[a,b,i_2,\ldots,i_n] + \sum_{\substack{g_1 + g_2 = g \\ J_1 \dot{\cup} J_2 = \{i_2,\ldots,i_n\}}} F_{g_1,1+|J_1|}[a,J_1] F_{g_2,1 + |J_2|}[b,J_2]\bigg)\,, \nonumber \\
\label{TReq2} 
\eea
where we observe that the Euler degree of the right-hand side is $\chi_{g,n} - 1$.  Therefore, \eqref{TReq2} is a recursion on $\chi_{g,n}$ which has at most one solution.

\vspace{0.2cm}

\noindent $\bullet$ \textbf{Proof of existence.} We now have to justify the existence of $F_{g,n} \in {\rm Sym}^n(V^*)$ satisfying \eqref{inieq}-\eqref{TReq2}, the non obvious part being the symmetry. A general holonomicity argument proving existence was given in \cite{KSTR}. Here we follow the pedestrian way of \cite{TRABCD}. First of all, \eqref{inieq} defines a symmetric $F_{0,3}$ because $A$ is fully symmetric in a quantum Airy structure. Then, we would like to take \eqref{TReq2} as recursive definition of $F_{g,n}$. This only makes sense if we can prove at each recursion step that, despite the fact that $i$ does not play the same role as $i_2,\ldots,i_n$ in \eqref{TReq2}, the result of the sum in the right-hand side is still symmetric when $i$ is permuted with one of the $i_m$. Note that the $i_2,\ldots,i_n$ do play a symmetric role, so it is enough to prove invariance under permutation of $i$ and $i_2$. Let us examine the two cases where $\chi_{g,n} = 2$, that is $(g,n) = (0,4)$ and $(2,1)$. We would like to define
$$
F_{0,4}[i,i_2,i_3,i_4] = B^i_{i_2,a}F_{0,3}[a,i_3,i_4] + B^i_{i_3,a}F_{0,3}[a,i_2,i_4] + B^i_{i_4,a}F_{0,3}[a,i_2,i_3]\,.
$$
Using $F_{0,3}[i,j,k] = A^i_{j,k}$ which is fully symmetric, this also reads
$$
F_{0,4}[i,i_2,i_3,i_4] = B^i_{i_2,a}A^{a}_{i_3,i_4} + B^i_{i_3,a}A^{i_2}_{a,i_4} + B^i_{i_4,a}A^{i_2}_{a,i_3}\,.
$$
We recognize the left-hand side of the first relation in \eqref{ABCeq}, therefore it is invariant if we exchange $i$ and $i_2$. We also would like to define
$$
F_{2,1}[i,i_2] = B^i_{i_2,a}D^a + \tfrac{1}{2}C^i_{a,b}F_{0,3}[a,b,i_2] = B^i_{i_2,a}D^a + \tfrac{1}{2}C^i_{a,b}A^{i_2}_{a,b}\,.
$$
We recognize the left-hand side of the $D$ relation \eqref{Deq}, therefore it is invariant if we exchange $i$ and $i_2$. We see that the conditions imposed on $(A,B,C,D)$ by the characterization of quantum Airy structures are essential in proving the symmetry of $F_{g,n}$.  

Now we examine the general case. Take $(g,n)$ such that $\chi_{g,n} > 2$, and assume we have proved full symmetry of $F_{g',n'}$ for all $\chi_{g',n'} < \chi_{g,n}$. Let $K = \{i_3,\ldots,i_n\}$. Let us define $F_{g,n}[i,j,K]$ by applying \eqref{TReq2} with first index $i$.
\begin{small}\bea
F_{g,n}[i,j,K] & = & B_{j,a}^iF_{g,n - 1}[a,K] + \sum_{k \in K} B_{k,a}^iF_{g,n - 1}[a,j,K^{(k)}] + \tfrac{1}{2}C_{a,b}^i F_{g - 1,n + 1}[a,b,j,K] \nonumber \\
&& + \sum_{\substack{h' + h'' = g \\ K' \dot{\cup} K'' = K}} C_{a,b}^i F_{h',2 + |K'|}[a,j,K'] F_{h'',1 + |K''|}[b,K'']\,. \nonumber
\eea
\end{small}
\noindent The resulting terms are in the range of the induction hypothesis, thus fully symmetric under permutation of $(j,i_3,\ldots,i_n)$. We use again \eqref{TReq2} with first index $j$, except for the term involving $B^i_{j,a}F_{g,n - 1}[a,K]$, for which we rather use \eqref{TReq2} with first index $a$. Denote $K^{(k)} := K \setminus\{k\}$ and $K^{(k,\ell)} := K\setminus\{k,\ell\}$.  We also implicitly use the full symmetry of $A$ and the symmetry of $C$ in its two lower indices. The result is
\begin{small}
\bea
&& F_{g,n}[i,j,K] \nonumber \\
& = & B_{j,a}^i\bigg\{\sum_{k \in K} B_{k,b}^a F_{g,n - 2}[b,K^{(k)}] + \tfrac{1}{2}C_{b,c}^aF_{g - 1,n}[b,c,K] + \sum_{\substack{h' + h'' = g \\ K' \dot{\cup} K'' = K}} \tfrac{1}{2}C_{b,c}^a F_{h',1+|K'|}[b,K']F_{h'',1 + |K''|}[c,K'']\bigg\} \nonumber \\
&& + \sum_{k \in K} B_{k,a}^i\bigg\{B_{a,b}^j F_{g,n - 2}[b,K^{(k)}] + \sum_{\ell \in K^{(k)}} B_{\ell,b}^j F_{g,n - 2}[b,a,K^{(k,\ell)}] + \tfrac{1}{2} C_{b,c}^j F_{g - 1, n}[b,c,a,K^{(k)}] \nonumber \\
&& + \sum_{\substack{h' + h'' = g \\ K' \dot{\cup} K'' = K^{(k)}}} \tfrac{1}{2} F_{h',2 + |K'|}[b,a,K'] F_{h'',1 + |K''|}[c,K''] \bigg\} \nonumber \\
&& + \tfrac{1}{2}C_{a,b}^i\bigg\{B_{a,c}^jF_{g - 1,n}[b,c,K] + B^j_{b,c}F_{g - 1,n}[a,c,K] \nonumber
\eea
\bea
&\phantom{=}&  + \sum_{k \in K} B_{k,c}^j F_{g - 1,n}[b,c,a,K^{(k)}] + \tfrac{1}{2}C_{c,d}^j F_{g - 2,n + 2}[c,d,a,b,K] \nonumber \\
&& + \sum_{\substack{h' + h'' = g - 1 \\ K' \dot{\cup} K'' = K}}  C_{c,d}^jF_{h',3+|K'|}[a,b,c,K']F_{h'',1 + |K''|}[d,K''] + C_{c,d}^j F_{h',2 + |K'|}[a,c,K']F_{h'',2 + |K''|}[b,d,K'']\bigg\} \nonumber \\
&& + \sum_{\substack{h' + h'' = g \\ K' \dot{\cup} K'' = K}} C_{a,b}^i F_{h'',1 + |K''|}[b,K'']\bigg\{B_{a,c}^j F_{h',1 + |K'|}[c,K'] + \sum_{k \in K'} B_{k,c}^j F_{h',1 + |K'|}[c,a,K^{(k)}] \nonumber \\
&& + \tfrac{1}{2}C_{c,d}^j F_{h' - 1,3 + |K'|}[c,d,a,K'] + \sum_{\substack{s + s' = h' \\ L \dot{\cup} L' = K'}} C_{c,d}^j F_{s,2 + |L|}[a,c,L]F_{s',1 + |L'|}[d,L'] + \delta_{h',0}\delta_{|K'|,1}A_{a,k'}^j\bigg\}\,. \nonumber
\end{eqnarray}
\end{small}
We now collect the various terms
\begin{small}
\begin{eqnarray}
&& F_{g,n}[i,j,K] \nonumber \\
&= & \sum_{k \in K} \textcolor{red}{F_{g,n - 2}[b,K^{(k)}] \Big(B_{j,a}^iB_{k,b}^a + B_{a,b}^jB_{k,a}^i + A_{a,k}^jC_{a,b}^i\Big)}  + \sum_{k \neq \ell \in K} F_{g,n - 2}[a,b,K^{(k,\ell)}]\,B_{k,a}^i B_{\ell,b}^j \nonumber \\
&& + \textcolor{red}{\tfrac{1}{2} F_{g - 1,n}[b,c,K]\big(B_{j,a}^iC_{b,c}^a + B_{a,c}^jC_{a,b}^i + B^j_{b,c}C^i_{a,b}\big)} + \sum_{k \in K} \tfrac{1}{2} F_{g - 1,n}[b,c,a,K^{(k)}]\big(B^i_{k,a}C^j_{b,c} + B^j_{k,c}C^i_{a,b}\big) \nonumber \\ 
&& + \tfrac{1}{4} F_{g - 2,n + 2}[a,b,c,d,K]\,C_{a,b}^i C_{c,d}^j + \textcolor{red}{\sum_{\substack{h' + h'' = g \\ K' \dot{\cup} K'' = K}} \tfrac{1}{2} F_{h',1+|K'|}[b,K']F_{h'',1+|K''|}[c,K'']\big(C_{b,c}^aB_{j,a}^i + 2C_{a,b}^iB_{a,c}^j\big)} \nonumber \\
&&+ \sum_{k \in K} \sum_{\substack{h' + h'' = g \\ K' \dot{\cup} K'' = K^{(k)}}}  \tfrac{1}{2} F_{h',2+|K'|}[b,a,K']F_{h'',1 + |K''|}[c,K'']\big(B_{k,a}^iC_{b,c}^j + B^j_{k,c}C^i_{a,b}\big) \nonumber \\
&& + \sum_{\substack{h' + h'' = g - 1\\ K' \dot{\cup} K'' = K}} \tfrac{1}{2}F_{h',3 + |K'|}[a,b,c,K'] F_{h'',1 + |K''|}[d,K'']\big(C_{c,d}^jC_{a,b}^i + C_{a,d}^iC_{c,b}^j\big) \nonumber \\
&& + \sum_{\substack{h' + h'' = g - 1\\ K' \dot{\cup} K'' = K}} \frac{1}{2}F_{h',2 + |K'|}[a,c,K']F_{h'',2 + |K''|}[b,d,K'']\,C_{a,b}^iC_{c,d}^j \nonumber \\
&& + \sum_{\substack{h'' + s + s' = g \\ K'' \dot{\cup} L \dot{\cup} L' = K}} F_{h'',1 + |K''|}[b,K'']F_{s,2 + |L|}[a,c,L]F_{s,1 + |L'|}[d,L']C_{c,d}^j C_{a,b}^i \,. \nonumber
\end{eqnarray}
\end{small}
The only terms which are not \textit{a priori} symmetric in $i$ and $j$ are highlighted in red. But the two last relations in \eqref{ABCeq} imply the symmetry of the first two red expressions. In the third red expression, $b$ and $c$ play a symmetric role -- provided we also exchange the role of $h',K'$ with the one of $h'',K''$. Also exploiting the symmetry of $C$, this red sum is also equal to its dissymetrized version
\begin{small}
$$
\textcolor{red}{\sum_{\substack{h' + h'' = g \\ K' \dot{\cup} K'' = K}} \tfrac{1}{2} F_{h',1+|K'|}[b,K']F_{h'',1+|K''|}[c,K'']\big(C_{b,c}^aB_{j,a}^i + C_{b,a}^iB_{a,c}^j + C^i_{c,a}B^j_{a,b}\big)\,.}
$$
\end{small}
This is symmetric owing to the third relation in \eqref{ABCeq}. So, $F_{g,n}[i,j,K]$ is fully symmetric, and by induction we proved Theorem~\ref{MainTh}.

\subsection{The structure of TR}
\label{Sgraph}
\textsc{tr} is the recursion on $\chi_{g,n}$ given by formula \eqref{inieq}-\eqref{TReq2}, which defines the coefficients $F_{g,n}[i_1,\ldots,i_n]$ of the free energy in terms of the coefficients $(A,B,C,D)$ of a quantum Airy structure. This formula has a graphical interpretation which makes it easy to remember.

\begin{definition}
Fix $g \geq 0$ and $n \geq 1$ such that $\chi_{g,n} > 0$. We form the set $\mathbb{G}_{g,n}$ of $\Gamma = (G,T)$ where
\begin{itemize}
\item[$\bullet$] $G$ is a trivalent graph with $n$ ordered leaves, and $b_1(G) = g$.
\item[$\bullet$] $T \subseteq G$ is a spanning tree which contains the first leaf (considered as the root), and not the other leaves.
\item[$\bullet$] the edges $e = \{v,v'\}$ of $G$ which are not in $T$ must connect parent vertices, \textit{i.e.} the common ancestor of $v$ and $v'$ in the rooted tree $T$ is either $v$ or $v'$.
\end{itemize}
An automorphism of $\Gamma$ is a permutation of ${\rm Edge}(G)$ which preserves the graph structure of $G$. We denote ${\rm aut}(\Gamma)$ the set of automorphisms. We denote $E'(\Gamma)$ the set consisting of leaves and of edges which are not loops. By convention, $\mathbb{G}_{0,1} = \mathbb{G}_{0,2} = \emptyset$. We insist that $G$ does not include the data of a cyclic order of edges/leaves incident at a vertex. Therefore, the number of automorphisms of a given $\Gamma$ is a power of $2$.
\end{definition}

\vspace{0.2cm}

\noindent $\bullet$ \textbf{Recursive description of $\mathbb{G}_{g,n}$.} For $\chi_{g,n} = 1$, we have
\begin{figure}[h!]
\begin{center}
\includegraphics[width=0.6\textwidth]{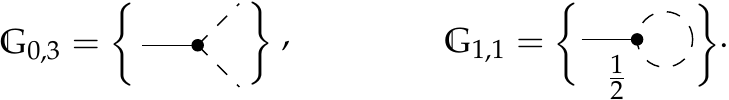}
\end{center}
\end{figure}

\vspace{-0.5cm}

In general, $\mathbb{G}_{g,n}$ has a recursive structure. Indeed, for $\chi_{g,n} \geq 2$, let us denote $\ell_1,\ldots,\ell_n$ the ordered leaves of $\Gamma \in \mathbb{G}_{g,n}$, and remove the vertex incident to $\ell_1$. It was also incident to two other edges/leaves $\{e_1,e_2\}$. We obtain $\Gamma'$ which can be
\begin{itemize}
\item[(\textbf{I})] a graph of $\mathbb{G}_{g,n - 1}$ if one of the $e_i$ is a leaf. The other edge $e_{\overline{i}}$ is then considered as the root of $\Gamma'$. 
\item[(\textbf{I}')] a graph in $\mathbb{G}_{g - 1,n + 2}$, with an arbitrary choice of first and second leaf to make.
\item[(\textbf{II})] a (non-ordered) disjoint union of $\Gamma'_1 \dot{\cup} \Gamma'_2$ where $\Gamma'_i \in \mathbb{G}_{g_i,1 + |J_i|}$ contains $e_i$ as the root, for a splitting of genera $g_1 + g_2 = g$ and a splitting $J_1 \dot{\cup} J_2$ of leaves of $\Gamma$ distinct from the first leaf. The ordering of $e_1$ and $e_2$, hence of $\Gamma'_1$ and $\Gamma'_2$ is arbitrary.
\end{itemize}
In the two last cases we have $|{\rm aut}(\Gamma)| = 2|{\rm aut}(\Gamma')|$. This is summarized in the picture below, where the arrow indicate the new ``first leaf'' of the connected components of $\Gamma'$.

\begin{figure}[h!]
\begin{center}
\includegraphics[width=0.7\textwidth]{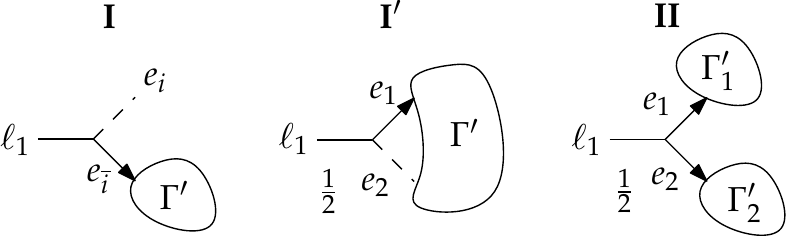}
\end{center}
\end{figure}

\noindent $\bullet$ \textbf{Weights.} Given coefficients $(A,B,C,D)$ of a quantum Airy structure, we define the weight $w(\Gamma,\gamma)$ of a graph $\Gamma \in \mathbb{G}_{g,n}$ equipped with a coloring $\gamma\,:\,E'(\Gamma) \rightarrow I$. We declare the base cases
$$
w(\Gamma_{0,3},c) = A^{\gamma(\ell_1)}_{\gamma(\ell_2),\gamma(\ell_3)},\qquad w(\Gamma_{1,1},c) = D^{\gamma(\ell_1)}\,,
$$
and then make a recursive definition for $\chi_{g,n} \geq 2$ using the previous decomposition. If we denote $\gamma'$ the restriction of $\gamma$ to $E'(\Gamma')$ and make a similar definition for $(\gamma'_i)_{i  = 1}^2$, we declare
\begin{itemize}
\item[(\textbf{I})] $w(\Gamma,\gamma) = B^{\gamma(\ell_1)}_{\gamma(e_i),\gamma(e_{\overline{i}})}\,w(\Gamma',\gamma')$.
\item[(\textbf{I}')] $w(\Gamma,\gamma) = C^{\gamma(\ell_1)}_{\gamma(e_1),\gamma(e_2)}\,w(\Gamma',\gamma')$.
\item[(\textbf{II})] $w(\Gamma,\gamma) = C^{\gamma(\ell_1)}_{\gamma(e_1),\gamma(e_2)}\,w(\Gamma'_1,\gamma'_1) w(\Gamma'_2,\gamma'_2)$.
\end{itemize}
These definitions are tailored so that the \textsc{tr} formula is equivalent to
\begin{lemma}
\label{TRgraph} For any $g \geq 0$, $n \geq 1$ and $i_1,\ldots,i_n \in I$, we have
$$
F_{g,n}[i_1,\ldots,i_n] = \sum_{\Gamma \in \mathbb{G}_{g,n}} \sum_{\substack{\gamma \in I^{E'(\Gamma)} \\ \gamma(\ell_j) = i_j}} \frac{w(\Gamma,\gamma)}{|{\rm aut}(\Gamma)|}\,.
$$ \hfill $\Box$
\end{lemma}

\subsection{A one-dimensional example}

The simplest example we can think of is $V = \mathbb{C}$. Any differential operator of the form
$$
L = \hbar \partial_{x} - \tfrac{1}{2}Ax^2 - \hbar Bx\partial_{x} - \tfrac{\hbar^2}{2}C\partial_{x}^2 - \hbar D
$$
forms a quantum Airy structure on the trivial Lie algebra $\mathbb{C}$. According to Lemma~\ref{TRgraph}, the partition function then enumerates graphs in $\mathbb{G}_{g,n}$, weighted by automorphisms, according to the number of vertices which received an $A$, $B$, $C$ or $D$ weight. If we specialize to $A = B = C = D = 1$, we get
$$
F = \sum_{g \geq 0} \sum_{n \geq 1} \frac{\hbar^{g - 1}x^n}{n!}\,|\mathbb{G}_{g,n}|,\qquad |\mathbb{G}_{g,n}| = \sum_{\Gamma \in \mathbb{G}_{g,n}} \frac{1}{|{\rm aut}(\Gamma)|}\,.
$$

On the other hand, the differential equation $L\cdot \exp(F) = 0$ is easy to solve. If we set
$$
\exp(F(x)) = \exp\big(\tfrac{1}{\hbar}(x - \tfrac{x^2}{2})\big)\tilde{Z}(y),\qquad y = \tfrac{1 - 2x - \hbar}{(2\hbar)^{2/3}}\,,
$$
we reduce this differential equation to the Airy differential equation
$$
\partial_{y}^2 \tilde{Z}(y) = y\tilde{Z}(y)\,.
$$
Taking into account the fact that $F(x) = O(x)$ and the initial condition $F_{0,3} = A = 1$, we deduce that $\tilde{Z}(y) = \frac{{\rm Bi}(y)}{{\rm Bi}(0)}$ where ${\rm Bi}$ is the Bairy function, whose asymptotic expansion when the variable $y \rightarrow +\infty$ is
$$
{\rm Bi}(y) = y^{-1/4}\exp\big(-\tfrac{2}{3}y^{3/2}\big)\Big(1 + \sum_{m \geq 1} \tfrac{6^m\Gamma(m + \frac{1}{6})\Gamma(m + \frac{5}{6})}{2\pi}\,\tfrac{y^{-3m/2}}{m!}\Big)\,.
$$
Note that logarithm of $\tilde{Z}(y)$ has an asymptotic expansion when $y \rightarrow \infty$ which becomes an element of $\mathcal{E}_{V}$ once we substitute $y = (2\hbar)^{-2/3}(1 - 2x - \hbar)$ and consider the result as a formal series in $\hbar$ and $x$. It is only in this sense that the Bairy function ``computes'' the generating series of trivalent graphs. This is actually a well-known result. For general $A,B,C,D$, one can still solve the differential equation $L\cdot \exp(F) = 0$, and the result is expressed in terms of the Whittaker-M function, see \cite{TRABCD}.

The appearance of the Airy differential equation here motivates the name ``quantum Airy structure''.

\subsection{Operations on quantum Airy structures}
\label{Sop}

The group $\exp(\mathcal{W}_{V}(\leq 2))$ acts on its Lie algebra $\mathcal{W}_{V}(\leq 2)$ by conjugation. So, if $L\,:\,V \rightarrow \mathcal{W}_{V}(\leq 2)$ is a quantum Airy structure, we can consider for $U \in \exp(\mathcal{W}_{V}(\leq 2))$ the new system of differential operators and the partition function it annihilates
\beq
\label{actionL} L \rightarrow \tilde{L} = ULU^{-1},\qquad Z \rightarrow \tilde{Z} = UZ\,.
\eeq
Yet, for a general $U$, $ULU^{-1}$ will not be a quantum Airy structure. Apart from $\hbar$-degrees, the main obstruction is that $L$ should contain as linear terms only derivations and no linear terms $x_i$. We shall describe three such operations that do preserve quantum Airy structures. When this is so, it is for a conjugated (thus isomorphic) Lie algebra structure on $V$. The resulting transformation $(A,B,C,D) \rightarrow (\tilde{A},\tilde{B},\tilde{C},\tilde{D})$ can be explicitly thanks to 

\begin{lemma}
\label{BCH} Let $(u_{i,j})_{i,j \in I}$ be a symmetric matrix, $V = (v_{i,j})_{i,j \in I}$ be a matrix, and $(t_i)_{i \in I}$ a vector. We have
\bea
\exp(-\tfrac{\hbar}{2} u_{a,b} \partial_{x_a}\partial_{x_b})t_ax_{a}\exp\big(\tfrac{\hbar}{2} u_{a,b}\partial_{x_a}\partial_{x_b}\big) & = & t_ax_a + \hbar u_{a,b}t_a\partial_{x_b} + \tfrac{\hbar}{2}u_{a,b}t_{a}t_{b}\,, \nonumber \\
\exp(-\tfrac{1}{2\hbar} u_{a,b}x_{a}x_{b}\big) \hbar t_{a}\partial_{x_a} \exp\big(\tfrac{1}{2\hbar} u_{a,b} x_{a}x_{b}\big) & = & \hbar t_{a}\partial_{x_a} - u_{a,b}t_a x_b + \tfrac{1}{2} u_{a,b}t_{a}t_{b}\,, \nonumber \\
\exp(-v_{a,b}x_{a}\partial_{x_b}) x_i \exp(v_{a,b}x_{a}\partial_{x_b}) & = & (e^{-V})_{a,i}x_a\,, \nonumber \\
\exp(-v_{a,b}x_{a}\partial_{x_b}) \hbar t_a\partial_{x_a} \exp(v_{a,b}x_{a}\partial_{x_b}) & = & (e^{V})_{i,a}\,\hbar \partial_{x_a}\,. \nonumber 
\eea
\end{lemma}

\begin{proof} Let $X,Y$ be elements of a Lie algebra such that $[X,[X,Y]]$ is central. The Baker-Campbell-Hausdorff formula in this special case implies
$$
\exp(-X)Y\exp(X) = Y - [X,Y] + \tfrac{1}{2}[X,[X,Y]]\,.
$$
We apply it to the Lie algebra $\mathcal{W}_{V}(\leq 2)$. With $X = \tfrac{\hbar}{2} u_{a,b}\partial_{x_a}\partial_{x_b}$ and $Y = t_ax_a$, we compute $[X,Y] = \hbar u_{a,b}t_{a}\partial_{x_b}$ and $[X,[X,Y]] = \hbar u_{a,b}t_{a}t_{b}$ which is indeed central, and this leads to the first claim. With $X = \tfrac{1}{2\hbar} u_{a,b}x_ax_b$ and $Y = \hbar t_a \partial_{x_a}$, we compute $[X,Y] = -u_{a,b}t_{a}x_b$ and $[X,[X,Y]] = -u_{a,b}t_{a}t_{b}$ which is central, and this leads to the second claim.

Now, we consider $X = v_{a,b}x_a\partial_{x_b}$ and introduce the vector $Y = (Y_i)_{i \in I}$ with $Y_i = x_i$. We compute $[X,Y_i] = v_{i,a}x_a$. Commuting further with $X$ does not give a central element, so we have to use another method. Let us work over $\mathbb{K}[[z]]$ instead of $\mathbb{K}$. Setting $G(z) = \exp(-zX)Y\exp(zX)$, we compute
$$
\partial_{z}G(z) = -\exp(-zX)[X,Y]\exp(zX) = -V^{T}G(z),\qquad G(0) = Y\,,
$$
which is solved by $G(z) = \exp(-zV^{T})Y$. Specializing to $z = 1$ replaces $\tilde{\mathbb{K}}$ by $\mathbb{K}$ and gives the third claim -- if $\dim V = \infty$, one must assume convergence. The fourth claim is proved in a similar way.
\end{proof} 

For convenience, we fix a basis $(e_i)_{i \in I}$ in which $L_i = \hbar\partial_{x_i} + \cdots$. Note that, if $\tilde{L} = U^{-1}LU$ is a quantum Airy structure, it only means that $U^{-1}L_iU = M_{i,a} \hbar \partial_{x_a} + \cdots$, the new operators of the form \eqref{Liform} must be $\tilde{L}_i = M^{-1}_{i,a}U^{-1}L_{a}U$.

\vspace{0.2cm}

\noindent$\bullet$ \textbf{Change of basis.} As the notion of quantum Airy structure is basis independent, any linear isomorphism $\Phi\,:\,V \rightarrow V$ induces a isomorphism $\tilde{\Phi}\,:\,\mathcal{W}_{V} \rightarrow \mathcal{W}_{V}$ and $\tilde{\Phi}^{-1} \circ L \circ \Phi$ is a new Airy structure. If we assume $\Phi = \exp(-V^{T})$, Lemma~\ref{BCH} tells us that this transformation is realized by
$$
U = \exp\big(v_{a,b} x_a\partial_{x_b}\big)\,,
$$
where $(v_{i,j})_{i,j \in I}$ is the matrix of $\Phi\,:\,V \rightarrow V$ in the basis $(e_i)_{i \in I}$. This amounts to
\bea
\tilde{A}^i_{j,k} & = & \Phi_{i,a}\Phi_{j,b}\Phi_{j,c}\,A^{a}_{b,c}\,, \nonumber \\
\tilde{B}^i_{j,k} & = & \Phi_{i,a}\Phi_{j,b}\Phi^{-1}_{c,k}\,B^a_{b,c}\,,\nonumber \\
\tilde{C}^i_{j,k} & = & \Phi_{i,a}\Phi^{-1}_{b,j}\Phi^{-1}_{c,k}\,C^{a}_{b,c}\,, \nonumber \\
\tilde{D}^i & = & \Phi_{i,a}D^a\,. \nonumber
\eea

\vspace{0.2cm}

\noindent $\bullet$ \textbf{Pure quadratic differential operators.} If $(u_{i,j})_{i,j \in I}$ is a symmetric matrix, the action of
$$
U = \exp\big(\tfrac{\hbar}{2} u_{a,b}\partial_{x_a}\partial_{x_b}\big)
$$ 
preserves quantum Airy structures. According to Lemma~\ref{BCH}, it amounts to
\beq
\label{utran}\left\{\begin{array}{rcl} \tilde{A}^i_{j,k} & = & A^i_{j,k} \\
\tilde{B}^i_{j,k} & = & B^i_{j,k} + A^i_{j,a}u_{a,k}  \\
\tilde{C}^i_{j,k} & = & C^i_{j,k} + B^i_{a,j}u_{a,k} + B^i_{a,k}u_{a,j} + A^i_{a,b}u_{a,j}u_{b,k} \\
\tilde{D}^i & = & D^i + \tfrac{1}{2}A^i_{a,b}u_{a,b}
\end{array}\,. \right.
\eeq

\vspace{0.2cm}

\noindent $\bullet$ \textbf{Translation.} Here we assume that there exists $r > 0$ such that for any $g \geq 0$, the series $F_{g}$ has a radius of convergence $r$ -- if $\dim V = \infty$, extra care is needed to make sense of this. We take $t \in V$ inside the radius of convergence. $Z(x + t) = \exp(\hbar^{-1} t_a\partial_{x_a})Z(x)$ cannot be a partition function of a quantum Airy structure, because it has $(0,1)$ term $\tfrac{1}{\hbar}(\partial_{t_a}\mathcal{F}_{0}(t))x_{a}$ and $(0,2)$ term $\tfrac{1}{2\hbar}(\partial t_{a}\partial_{t_b} \mathcal{F}_{0}(t))x_{a}x_{b}$. If we remove these two terms, \textit{i.e.} act with
$$
U = \exp\big(-\tfrac{1}{\hbar}(\partial_{t_a}\mathcal{F}_{0}(t))x_{a} - \tfrac{1}{2\hbar}(\partial_{t_a}\partial_{t_b}\mathcal{F}_{0}(t))x_{a}x_{b}\big) \exp\big(t_a\partial_{x_a}\big)\,,
$$
then we obtain a new quantum Airy structure. By consistency it must have
$$
\tilde{A}^i_{j,k} = \partial_{t_i}\partial_{t_j}\partial_{t_k} \mathcal{F}_{0}(t),\qquad \tilde{D}^i = \partial_{t_i} \mathcal{F}_{1}(t)\,.
$$
The transformation of $B$ and $C$ is computed by changing directly $x \rightarrow x + t$
\bea
\tilde{B}^i_{j,k} & = & (M^{-1})_{i,a}\big(B^a_{j,k} + C^a_{k,b} \partial_{t_j}\partial_{t_b} \mathcal{F}_{0}(t)\big)\,, \nonumber \\
\tilde{C}^i_{j,k} & = & (M^{-1})_{i,a} C^a_{j,k}\,, \nonumber \\
\tilde{D}^i & = & (M^{-1})_{i,a}D^a\,. \nonumber
\eea
where $M_{i,j} = \delta_{i,j} - B^i_{j,a}t_{a} - C^i_{j,a}\partial_{t_a} \mathcal{F}_0(t)$.

It can also be checked directly that if $(A,B,C,D)$ is solution of the six relations \eqref{fAeq}-\eqref{Deq}-\eqref{ABCeq}, then for each  of the transformations above, $(\tilde{A},\tilde{B},\tilde{C},\tilde{D})$ is a solution of the six relations as well. 

\vspace{0.2cm}

\noindent $\bullet$ \textbf{Direct sum.} There is a last obvious operation we can do if we dispose of a family of quantum Airy structure $L^{(\alpha)}\,:\,V^{()} \rightarrow \mathcal{W}_{V^{(\alpha)}}$ indexed by $\alpha$, with partition function $Z^{(\alpha)}$: the direct sum $\bigoplus_{\alpha} L^{(\alpha)}$ is a quantum Airy structure on $\bigoplus_{\alpha} V^{(\alpha)}$, with partition function $\prod_{\alpha} Z^{(\alpha)}$.

\subsection{Bonus: from classical to quantum Airy structures}

The motivation of Kontsevich and Soibelman to introduce quantum Airy structures comes from an interpretation of \textsc{tr} as a quantization of quadratic Lagrangians in the symplectic vector space $T^*V$. 

\vspace{0.2cm}

\noindent $\bullet$ \textbf{Classical Airy structures.} For simplicity we assume $V$ finite-dimensional and $\mathbb{K} = \mathbb{R}$ or $\mathbb{C}$. Denote $(x_i)_{i \in I}$ linear coordinates on $V$, and $(y_i)_{i \in I}$ the corresponding linear coordinates on $V^*$. We equip $T^*V$ with its canonical symplectic structure $\omega$, and denote $\{\cdot,\cdot\}$ the corresponding Poisson bracket
$$
\forall i,j \in I,\qquad \{x_i,x_j\} = \{y_i,y_j\} = 0,\qquad \{y_i,x_j\} = \delta_{i,j}\,.
$$

\begin{definition}
 A classical Airy structure is the data of scalars $A^i_{j,k} = A^i_{k,j}$, $B^i_{j,k}$, $C^i_{j,k}$ and $f_{i,j}^k = -f_{i,j}^k$ such that the family of functions $(\lambda_i)_{i \in I}$ on $T^*V$ defined 
\begin{equation}
\label{lambdai} \lambda_i =  y_i - \tfrac{1}{2}A^i_{a,b}x_ax_b - B^i_{a,b}x_ay_b - \tfrac{1}{2}C^i_{a,b}y_{a}y_{b}
 \end{equation}
 satisfy $\{\lambda_i,\lambda_j\} = f_{i,j}^a \lambda_{a}$ for all $i,j \in I$.
\end{definition}
A classical Airy structure is thus characterized by $(A,B,C,f)$ satisfying the five relations \eqref{fAeq} and \eqref{ABCeq}.

\begin{lemma}
If $(A,B,C,f)$ is a classical Airy structure, then the subvariety $\Lambda = \{(x,y) \in T^*V\,\,:\,\,\forall i \in I,\,\,\lambda_i(x,y) = 0\}$ is a Lagrangian in a neighborhood of $0 \in T^*V$.
\end{lemma}
\begin{proof} Let $X^{(i)}$ be the hamiltonian vector field of $\lambda_i$. The function $\{\lambda_i,\lambda_j\} = f_{i,j}^a\lambda_{a}$ vanishes on $\Lambda$. This is also equal to $\{X^{(i)},X^{(j)}\} = \dd\lambda_i(X^{(j)})$, which thus vanishes on $\Lambda$. Therefore, the restriction of $X^{(j)}$ to $\Lambda$ is a tangent vector field to $\Lambda$. In a neighborhood of $0$, $(\dd\lambda_i)_{i \in I}$ are linearly independent. Therefore, $\Lambda$ restricted to this neighborhood is a manifold of dimension $\tfrac{1}{2}\dim V$, whose tangent space is spanned by $(X^{(i)})_{i \in I}$. Hence it is Lagrangian.
\end{proof}

\vspace{0.2cm}

\noindent $\bullet$ \textbf{Quantization.} The Weyl algebra $\mathcal{W}_{V}$ provides a canonical deformation quantization of the Poisson manifold $T^*V$, \textit{i.e.} an algebra over $\mathbb{K}[[\hbar]]$ with a linear projection map $p\,:\,\mathcal{W}_{V} \rightarrow \mathbb{K}[T^*V]$ such that
$$
[X,Y] = \hbar \{p(X),p(Y)\} + O(\hbar^2)\,.
$$
Quantizing a Lagrangian $\Lambda \subset T^*V$ then means constructing a $\mathcal{W}_{V}$-module $\hat{\Lambda}$, which in the limit $\hbar \rightarrow 0$ becomes the ring of functions on $\Lambda$, \textit{i.e.} $\mathbb{K}[T^*V]/\{\lambda_i = 0\}$. Here, we are only considering Lagrangians of $T^*V$ defined by quadratic equations of the form \eqref{lambdai}. Finding such Lagrangians is a non-trivial task, as one has to solve in particular the three quadratic relations \eqref{ABCeq}. But if we imagine that we have found one $\Lambda$, completing $(A,B,C)$ to a quantum Airy structure $(A,B,C,D)$ is fairly easy as it just requires solving the affine equation \eqref{Deq} for $D$ (see the remark at the end of this lecture). In fact, if we want to quantize the monomials $x_iy_j$ in the $\lambda$s, we have two choices $\hbar\,x_i\partial_{x_j}$ and $\hbar\,\partial_{x_j}x_i$, which differ by a constant times $\hbar$. The choice of $D$ can be seen as a prescription for quantization of those monomials, but it is not arbitrary, as we desire that the quantization $\lambda_i \rightsquigarrow L_i$ lifts the Poisson commutation relation to commutation relations. The resulting constraint on $D$ is \eqref{Deq}. If we find such a quantum Airy structure $L\,:\,V \rightarrow \mathcal{W}_{V}$, we then get a $\mathcal{W}_{V}$-module $\hat{\Lambda} = \mathcal{W}_{V}/L.\mathcal{W}_{V}$ which quantizes the Lagrangian $\Lambda$ in the above sense.

The space $\exp(\mathcal{E}_{V})$ of formal functions on $V$ defined in \eqref{Edef} is another $\mathcal{W}_{V}$-module. We can then consider the space ${\rm Hom}_{\mathcal{W}_{V}}\big(\mathcal{W}_{V}/L.\mathcal{W}_{V},\exp(\mathcal{E}_{V})\big)$, which is isomorphic as a vector space to the space of solutions of $L(v).Z= 0$ for all $v \in V$. The partition function of the quantum Airy structure  is an element of this vector space, which generates it under the action of $\exp(\mathcal{W}_{V}(\leq 2))$. Sometimes, such a $Z$ is also called ``wave function''.
\vspace{0.2cm}

\noindent $\bullet$ \textbf{A representation theory perspective on classical Airy structures.} We revisit the five relations in \eqref{fAeq} and \eqref{ABCeq} which characterize classical Airy structures. If $L\,:\,V \rightarrow \mathcal{W}_{V}(\leq 2)$ is an arbitrary linear map, we can study the commutator map $\rho\,:\,V \rightarrow {\rm End}(T^*V)$ obtained by considering $T^*V \subset \mathcal{W}_{V}(1)$ and
$$
\forall w \in T^*V,\qquad \rho(v)(w) = \hbar^{-1} [L(v),w]\,.
$$
We note that $\rho$ is not sensitive to the presence of constants in $L$, it only depends on the coefficients $(A,B,C)$. Let us compute its matrix in a basis in which $L_i$ has the form \eqref{Liform}, that is the basis  $x_i$ for $V^* \subset \mathcal{W}_{V}(1)$, and $\hbar\partial_{x_i}$ for $V \subset \mathcal{W}_{V}(1)$. We find
\beq
\label{rhoi}\rho_i = \left(\begin{array}{cc} -B^i & A^i \\ C^i & (B^i)^{T} \end{array}\right)\,,
\eeq
where we decomposed $T^*V = V^* \oplus V$, and $X^i$ is the matrix $(X^i_{j,k})_{j,k \in I}$. We therefore find that $\rho(v)$ belongs to  $\mathfrak{sp}(T^*V)$, the Lie algebra of symplectic matrices in $T^*V$ -- equipped with its canonical symplectic structure.
\begin{lemma}
A classical Airy structure is equivalent to the data of a Lie algebra structure on $V$ and a Lagrangian embedding $\varphi\,:\,V \rightarrow T^*V$, together with a Lie algebra representation $\rho\,:\,V \rightarrow \mathfrak{sp}(T^*V)$ such that
\beq
\label{0tors} \forall v_1,v_2 \in V,\qquad \rho(v_1)(\varphi(v_2))- \rho(v_2)(\varphi(v_2)) = \varphi([v_1,v_2])\,.
\eeq
\end{lemma}
\begin{proof} $\rho$ is a representation if and only if we have $\rho([v_1,v_2]) = [\rho(v_1),\rho(v_2)]$ for any $v_1,v_2 \in V$. Let us choose a basis $(e_i)_{i \in I}$ of the abstract vector space $V$, which we transport with $\varphi$ to a basis still denoted $(e_i)_{i \in I}$ of $V \subset T^*V$, in which $L(e_i)$ takes the form \eqref{rhoi}. Making explicit the condition that $\rho$ is a representation and that \eqref{0tors} holds, yields exactly the five conditions \eqref{fAeq}-\eqref{ABCeq}.
\end{proof}

By analogy with the condition for Riemannian connections acting on vector fields, we may call \eqref{0tors} a zero-torsion condition.

\vspace{0.2cm}

\noindent $\bullet$ \textbf{Quantizing Airy structures.} To lift a classical Airy structure $(A,B,C)$ to a quantum Airy structure, we need a $D \in {\rm Hom}_{\mathbb{K}}(V,\mathbb{K})$ solution of \eqref{Deq}. As it is an affine relation, if $D_0$ is a solution, then $D$ is another solution if and only if $f_{i,j}^a(D^a - D^a_0) = 0$. In other words, $D$ is another solution iff
$$
(D - D_0)([V,V])= 0\,.
$$
To conclude, let us say that when $V$ is finite-dimensional, one can check that $D_{0}^i = \tfrac{1}{2} \sum_{k \in I} B^i_{k,k}$ is always a solution of \eqref{Deq}. Indeed, in this case \eqref{Deq} is implied by taking $k = \ell$ and summing over $k \in I$ in the second relation of \eqref{ABCeq}.

\newpage
\lecture[2h]{10}{05}{2017}

\section{Two-dimensional topological quantum field theories}

\subsection{Definitions}

\noindent $\bullet$ \textbf{The category $\textbf{Bord}_{2}$.} Its objects are finite disjoint union of topological, compact, oriented, connected $1$-manifolds. If $B_1,B_2$ are two objects, the morphisms $B_1 \rightarrow B_2$ are equivalence classes of compact oriented surfaces $\Sigma$, with oriented boundaries\footnote{What we call boundary throughout the text is, more precisely speaking, a boundary component.} such that the set of boundaries of $\Sigma$ whose orientation disagree with the orientation induced from $\Sigma$ is $\partial_{-} \Sigma = B_1$, and the set of boundaries whose orientation agree is $\partial_{+}\Sigma = B_2$. Here we say that $\Sigma$ is equivalent to $\Sigma'$ if there exists a diffeomorphism $\varphi\,:\,\Sigma \rightarrow \Sigma'$ preserving the orientation of $\Sigma$ and of the boundaries. If $B_1,B_2,B_3$ are three objects, $[\Sigma_{1}]$ is a morphism from $B_1 \rightarrow B_2$ and $[\Sigma_{2}]$ is a morphism from $B_2$ to $B_3$, we choose a parametrization (\textit{i.e.} a diffeomorphism to the standard $\mathbb{S}^1$) of each connected component of $\partial_{+}\Sigma_1$ and of $\partial_{-}\Sigma_2$, and form a surface $\Sigma_{1,2}$ by glueing together $\partial_{+}\Sigma_1$ and $\partial_{-}\Sigma_2$ such that the parametrizations match. The choice of parametrization was necessary to perform a glueing, but the resulting equivalence class $[\Sigma_{1,2}]$ does not depend on this choice, neither on the choice of representatives $\Sigma_1$ and $\Sigma_2$. This allows the definition of the composition of morphisms $[\Sigma_2] \circ [\Sigma_1] = [\Sigma_{1,2}]$. The disjoint union induces a monoidal structure on $\textbf{Bord}_{2}$, for which the empty disjoint union is the unit object.

The morphisms in this category are generated under composition by basic surfaces of genus $0$, with $n \leq 3$ boundaries -- which can be either negatively or positively oriented. Note that the cylinder with one negatively oriented boundary, and one positively oriented boundary, represents the identity morphism of the object $\mathbb{S}^1$ appearing at the boundaries.

\vspace{0.2cm}

\noindent $\bullet$ \textbf{The category $\textbf{Vect}_{\mathbb{K}}$.} Its objects are finite-dimensional $\mathbb{K}$-vector spaces, and its morphisms are $\mathbb{K}$-linear maps. The tensor product induces a monoidal structure on $\textbf{Vect}_{\mathbb{K}}$, for which $\mathbb{K}$ is the unit object.

\vspace{0.2cm}

\noindent $\bullet$ \textbf{2d TQFTs.} A two-dimensional topological quantum field theory is a monoidal functor $\mathcal{F}\,:\,\textbf{Bord}_{2} \rightarrow \textbf{Vect}_{\mathbb{K}}$. This axiomatization is due to Atiyah \cite{At2}, and can be done in any dimension $d$, using a category $\textbf{Bord}_{d}$ whose objects are $(d - 1)$-manifolds and morphisms are equivalence classes of $d$-manifolds.

\vspace{0.2cm}

\noindent $\bullet$ \textbf{Frobenius algebras.} A Frobenius algebra $\mathcal{A}$ is a finite-dimensional $\mathbb{K}$-vector space equipped with the structure of a commutative associative algebra $\mu\,:\,\mathcal{A} \otimes \mathcal{A} \rightarrow \mathcal{A}$ with a unit $\mathbf{1}\,:\,\mathbb{K} \rightarrow \mathcal{A}$, and a non-degenerate symmetric pairing $b\,:\,\mathcal{A} \otimes \mathcal{A} \rightarrow \mathbb{K}$ such that the product is invariant
$$
\forall a_1,a_2,a_3 \in \mathcal{A},\qquad b(\mu(a_1 \otimes a_2),a_3) = b(a_1,\mu(a_2 \otimes a_3))\,.
$$

\subsection{Correspondence with Frobenius algebras}

An important idea in the mathematical approach to quantum field theories is that manifolds and their glueings give rise to algebraic structures, where algebraic relations express diffeomorphisms between manifolds. The following result has been folklore knowledge, and it seems to have been written in detail for the first time in \cite{Abrams}. The negatively oriented boundaries are drawn to the left, and the positively oriented ones to the right.

\begin{theorem}
A \textsc{2d tqft} determines a Frobenius algebra structure on $\mathcal{F}(\mathbb{S}^1) = \mathcal{A}$ \textit{via} the formulas

\vspace{-0.2cm}

\begin{figure}[h!]
\begin{center}
\includegraphics[width=0.4\textwidth]{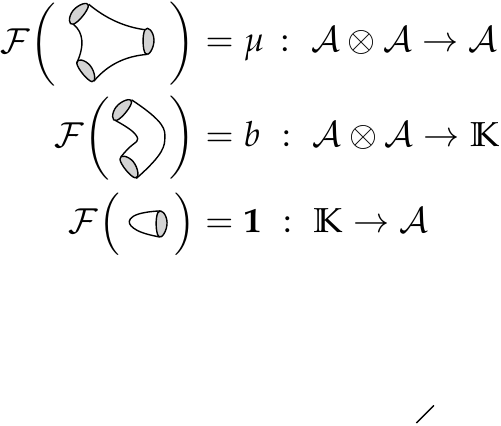}
\end{center}
\end{figure}

\vspace{-2cm}

\noindent Conversely, for any Frobenius algebras, these formulas determine a unique \textsc{2d tqft}.
\end{theorem}

\noindent \textit{Sketch of proof.} Let $P_{m,n}$ a morphism in $\textbf{Bord}_{2}$ which corresponds to a sphere with $m$ negatively oriented boundaries and $n$ positively oriented boundaries, and for which the boundary are parametrized by the standard $\mathbb{S}^1$. We can always find an automorphism of $P_{2,1}$ which exchanges the two negatively oriented boundaries, therefore $\mu$ is commutative. If $\{i,j\} \subset \{1,2,3\}$, let $\mu_{i,j}$ be the map $\mu$ acting on the $i$-th and $j$-th copy of $\mathcal{A}$ in $\mathcal{A}^{\otimes 3}$. Associativity is the identity $\mu \circ \mu_{1,2} = \mu \circ \mu_{2,3}$, and it comes from the diffeomorphism depicted below.

\begin{figure}[h!]
\begin{center}
\includegraphics[width=0.35\textwidth]{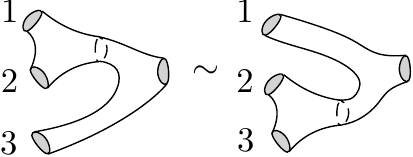}
\end{center}
\end{figure}

By writing diffeomorphisms between surfaces glued in different ways, one can prove likewise that $\mathbf{1}$ is a unit and $b$ is symmetric. The only non-obvious property is the non-degeneracy of $b$. Let $\Delta = \mathcal{F}(P_{0,2})\,:\,\mathbb{K} \rightarrow \mathcal{A} \otimes \mathcal{A}$. The diffeomorphism between the cylinder $P_{1,1}$ and the glued surface in the drawing implies

\vspace{-0.2cm}

\begin{figure}[h!]
\begin{center}
\includegraphics[width=0.8\textwidth]{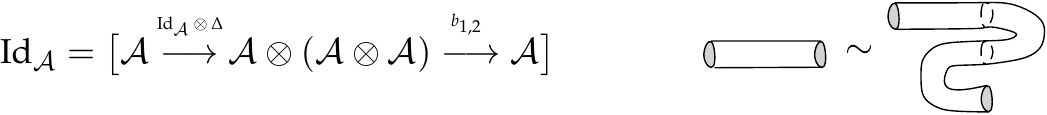}
\end{center}
\end{figure}

\vspace{-0.5cm}

Let us choose a basis $(e_i)_{i \in I}$ of $\mathcal{A}$, and write $\Delta = \sum_{i,j \in I} u_{i,j} e_i \otimes e_j$. Computing the two sides in the above equality in this basis yield
$$
\delta_{i,j} = \sum_{k \in I} b(e_i,e_k)\,u_{k,j}\,.
$$
Therefore, the matrix $\big(b(e_i,e_j)\big)_{i,j \in I}$ is invertible, with inverse $(u_{j,i})_{i,j \in I}$, hence $b$ is non-degenerate.

Conversely, we put $\mathcal{F}(B) = \mathcal{A}^{\pi_0(B)}$ for any object $B$ in $\textbf{Bord}_{2}$. We will describe in the next paragraph the multilinear maps $\mathcal{F}([\Sigma])\,:\,\mathcal{A}^{\pi_0(\partial_{-}\Sigma)} \rightarrow \mathcal{A}^{\pi_0(\partial_{+}\Sigma)}$. However, to specify the \textsc{tqft}, it is enough to assign multilinear maps to the surfaces which generate $\textbf{Bord}_{2}$ for the composition, and satisfying all glueing/diffeomorphism relations between those surfaces. $P_{n,m}$ with $n + m = 3$ is such a generating set of morphisms, and one can show that all glueing/diffeomorphism relations  between these surfaces are already used in the characterization of a Frobenius algebra. Then, as the non-degeneracy of the pairing gives a natural isomorphism between $\mathcal{A}$ and $\mathcal{A}^*$, the values $\mathcal{F}(P_{0,1}) = \mathbf{1}$, $\mathcal{F}(P_{2,0}) = b$ and $\mathcal{F}(P_{2,1}) = \mu$ specify the \textsc{2d tqft}.  \qed

\subsection{Computing the \textsc{tqft} amplitudes}
\label{TQFTa}
\noindent $\bullet$ \textbf{From the TQFT rules.} Let $\mathcal{F}\,:\,\textbf{Bord}_{2} \rightarrow \mathbf{Vect}_{\mathbb{K}}$ be a \textsc{2d tqft}. For $2g - 2 + n \geq 1$, let $\Sigma_{g,n}$ be a compact, oriented, topological surface of genus $g$ with $n$ negatively oriented boundaries. We describe how to compute the \textsc{tqft} amplitude $\mathcal{F}(\Sigma_{g,n}) \in (\mathcal{A}^*)^{\otimes n}$. We are going to use repeatedly the canonical identification between $\mathcal{A}$ and $\mathcal{A}^*$ given by the pairing. In particular, we have an element $\tilde{\mu} \in (\mathcal{A}^{*})^{\otimes 3}$ which represents the product in $\mathcal{A}$, and $\tilde{b}\,:\,\mathcal{A}^* \otimes \mathcal{A}^* \rightarrow \mathbb{K}$ which represents the pairing. The axioms of a Frobenius algebra imply that $\tilde{\mu}$ and $b$ are fully symmetric under permutation of the factors of $\mathcal{A}^*$. Let us choose a pair of pants decomposition $\mathbf{P} = (P_{v})_{v \in V}$ of $\Sigma_{g,n}$, and let $E$ the set of $\{v,v'\}$ such that $P_{v}$ and $P_{v'}$ are glued together in this decomposition. We form
$$ 
\bigg(\bigotimes_{\{v,v'\} \in E} \tilde{b}_{v,v'} \bigg)\circ \bigg(\bigotimes_{v \in V} \tilde{\mu}\bigg)
$$ 
where we have applied to the product of $\tilde{\mu}$ the pairing $\tilde{b}\,:\,\tilde{\mathcal{A}} \otimes \mathcal{A}^* \rightarrow \mathbb{K}$ acting on a copy of $\mathcal{A}^*$ in the $v$-th factor and a copy of $\mathcal{A}^*$ in the $v'$-th factor, for each $\{v,v'\} \in E$. Which copy of $\mathcal{A}^*$ in the $v$-th and $v'$-th factor and which order one chooses for $v$ and $v'$ is irrelevant as $\tilde{\mu}$ and $\tilde{b}$ are fully symmetric. There remains only $n$ factors of $\mathcal{A}^*$ which have not been paired, and the result belonging to $(\mathcal{A}^*)^{\otimes n}$ is the desired $\mathcal{F}(\Sigma_{g,n})$.

If $\mathbf{P}$ and $\mathbf{P}'$ are two pairs of pants decomposition, there exists an orientation preserving diffeomorphism of $\Sigma_{g,n}$ which takes $\mathbf{P}$ to $\mathbf{P}'$. So, by definition of a \textsc{2d tqft}, the result of this procedure does not depend on the choice of a pair of pants decomposition.

\vspace{0.2cm}

\noindent $\bullet$ \textbf{From quantum Airy structures.} Let $\mathcal{A}$ be a Frobenius algebra. We denote $\psi\,:\,\mathcal{A} \rightarrow \mathbb{K}$ the linear form defined by $\psi(x) = b(\mathbf{1},x)$. Let $(e_i)_{i \in I}$ and $(e_i^*)_{i \in I}$ be basis of $\mathcal{A}$ such that $b(e_i,e_j^*) = \delta_{i,j}$, and $H = \sum_{i \in I} e_i \cdot e_i^*$.

\begin{lemma}
The following data
\bea
A^i_{j,k} & = & \psi(e_i^*\cdot e_j^* \cdot e_k^*)\,, \nonumber \\
B^i_{j,k} & = & \psi(e_i^*\cdot e_j^*\cdot e_k)\,,  \nonumber \\
C^i_{j,k} & = & \psi(e_i^*\cdot e_j \cdot e_k)\,, \nonumber \\
D^i & = & \psi(e_i^*\cdot H)\,, \nonumber
\eea
defines a quantum Airy structure on the abelian Lie algebra $\mathcal{A}$. Besides, for $2g - 2 + n > 0$ the amplitudes of the \textsc{2d tqft} associated to $\mathcal{A}$ are computed by the \textsc{tr} amplitudes of this quantum Airy structure through
\beq
\label{Cc}\mathcal{F}(\Sigma_{g,n}) = |\mathbb{G}_{g,n}|F_{g,n}\,.
\eeq
\end{lemma}
In other words, the tensors $(A,B,C)$ in this quantum Airy structure all represent the product in $\mathcal{A}$.

\begin{proof} For simplicity, let us assume that $e_i$ is an orthonormal basis. Orthogonal basis always exist, but to obtain an orthonormal basis we may have to replace $\mathbb{K}$ by an extension $\tilde{\mathbb{K}}$. The final result however, is insensitive to this detail. In this case, we have $e_i^* = e_i$ and 
$$
A^i_{j,k} = B^i_{j,k} = C^i_{j,k}\,,
$$
which is fully symmetric. The three relations \eqref{ABCeq} are then identical, and we just need to check one of them. In the IHX representation, each vertex is an $A$ and the ordering of the edges at the vertices is irrelevant. When we permute the indices $i$ and $j$, the H-graph is sent to itself, while the I- and X-graph are swapped. Therefore, \eqref{ABCeq} is satisfied in a rather simple way. Then, by full symmetry of $B$, we have $f_{i,j}^{k} = 0$: the underlying Lie algebra is abelian. As a consequence, the first term in \eqref{Deq} involving $D$ is automatically symmetric, and as $A^i_{j,k} = C^i_{j,k}$ \eqref{Deq} is satisfied for any $D \in {\rm Hom}_{\mathbb{K}}(\mathcal{A},\mathbb{K})$.

For $2g - 2 + n = 1$, there is a single graph in $\mathbb{G}_{g,n}$, and it has trivial automorphism group. By construction, $F_{0,3} = A = \mathcal{F}(\Sigma_{0,3})$ represents the product, and we choose $D^i = \mathcal{F}(\Sigma_{1,1})(e_i)$ so that \eqref{Cc} holds for $(g,n) = (1,1)$. It is more explicitly computed by decomposing the torus with one boundary as a self-glued pair of pants.

\vspace{-0.2cm}

\begin{figure}[h!]
\begin{center}
\includegraphics[width=0.2\textwidth]{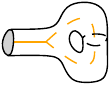}
\end{center}
\end{figure}

\vspace{-0.4cm}

\noindent The rules given above then give $\mathcal{F}(\Sigma_{1,1})(e_i) = \psi(e_i H)$ with $H = \sum_{j \in I} e_j\cdot e_j$.

Now consider $2g - 2 + n \geq 2$. To any $\Gamma \in \mathbb{G}_{g,n}$, we can associate a pair of pants decomposition $\mathbf{P}_{\Gamma}$ of $\Sigma_{g,n}$ as in the picture below. A coloring of $\Gamma$ is the assignment of an element of the basis to each edge/leaf, and since $A$, $B$ and $C$ are equal when written in the orthonormal basis, the contribution of each trivalent vertex $v$ to the weight is always the product of all the vectors carried by edges incident to that vertex -- when represented as an element of $\mathcal{A}^{\otimes 3} \rightarrow \mathbb{K}$. Moreover, summing over the colors of the edges amounts to perform a pairing because
$$
\forall x,y \in \mathcal{A},\qquad \sum_{a} \psi(x e_a)\psi(e_a y) = \psi(xy)\,,
$$
when we use the orthonormal basis $(e_i)_{i \in I}$. And, when $\Gamma$ has a loop, there is a self-glued pair of pants, which we can consider as a torus with one boundary, having by definition a contribution $D^i$ if the loop is color by $i$. 

\begin{figure}[h!]
\begin{center}
\includegraphics[width=0.4\textwidth]{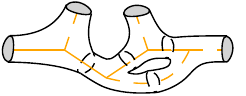}
\end{center}
\end{figure}

\vspace{-0.3cm}

\noindent Comparing with the \textsc{tqft} rules, we find that the sum over all colorings of $\Gamma$ which assigns the colors $(i_1,\ldots,i_n)$ to the $n$-leaves coincide with the \textsc{tqft} amplitude computed from the pair of pants decomposition $\mathbf{P}_{\Gamma}$, evaluated at $e_{i_1}\otimes \cdots \otimes e_{i_n}$. Since this amplitude does not depend on the pair of pants decomposition, the number of graphs (weighted by automorphisms) factorizes in the computation of the \textsc{tr} amplitude
\bea
F_{g,n} & = & \sum_{\Gamma \in \mathbb{G}_{g,n}} \sum_{\substack{\gamma \in I^{E'(\Gamma)} \\ \gamma(\ell_j) = i_{j}}} \frac{w(\Gamma,\gamma)}{|{\rm aut}(\Gamma)|} \nonumber \\
& = & \sum_{\Gamma \in \mathbb{G}_{g,n}} \frac{\mathcal{F}(\Sigma_{g,n})(e_{i_1} \otimes \cdots \otimes e_{i_n})}{|{\rm aut}(\Gamma)|}\,= |\mathbb{G}_{g,n}|\,\mathcal{F}(\Sigma_{g,n})(e_{i_1}\otimes \cdots \otimes e_{i_n})\,. \nonumber
\eea
\end{proof}

\section{Intersection theory on the moduli space of curves}

\subsection{A quantum Airy structure on the loop space}

\noindent $\bullet$ \textbf{Definition.} We describe a quantum Airy structure on $V = z\mathbb{C}[[z^2]]$. We take as basis of $V$
$$
\xi_k^* = \frac{z^{2k + 1}}{(2k + 1)!!}\,,
$$ 
indexed by $k \in \mathbb{N}$. We also introduce
\beq
\label{xidef} \xi_k = \frac{(2k + 1)!!}{z^{2k + 2}}\,\dd z\,,
\eeq
and take $\theta \in z^{-2}\mathbb{C}[[z^2]].(\dd z)^{-1}$ which we decompose as
$$
\theta = \sum_{r \geq -1} \theta_{r}\,z^{2r}\,(\dd z)^{-1}\,,
$$
with the convention $\theta_{r} = 0$ for $r \leq -2$.
\begin{theorem}
\label{thAiry} The following data
\beq
\label{ABCloop} \begin{array}{rclcl}
A^i_{j,k} & = & \Res_{0}\big(\xi_i^*\cdot \dd \xi_j^* \cdot \dd\xi_k^*\cdot \theta\big) & = & \delta_{i,j,k,0}\theta_{-1}\,,  \\
B^i_{j,k} & = & \Res_{0}\big(\xi_i^*\cdot \dd\xi_j^* \cdot \xi_k \cdot \theta\big) & = & \tfrac{(2k + 1)!!}{(2i + 1)!!(2j + 1)!!}\,(2j + 1)\,\theta_{k - i - j}\,, \\
C^i_{j,k} & = & \Res_{0}\big(\xi_i^*\cdot \xi_j \cdot \xi_k \cdot \theta\big) & = &\tfrac{(2k + 1)!!(2j + 1)!!}{(2i + 1)!!}\,\theta_{k + j + 1 - i}\,, \\
D^i & = & \frac{\theta_{-1}}{24}\delta_{i,0} + \frac{\theta_0}{8}\delta_{i,1}\,, &&
\end{array}
\eeq
defines a quantum Airy structure on $V \cong \bigoplus_{k \geq 0} \mathbb{C}.\xi_k^*$. If we denote $r_0 = \min\{r \geq -1\,\,:\,\,\theta_{r} \neq 0\}$, the Lie algebra structure it induces on $V$ is isomorphic to a sub-Lie algebra spanned by $(\mathcal{L}_{k})_{k \geq r_0}$ of the Lie algebra $\mathfrak{Witt}_+$, which is spanned by $(\mathcal{L}_{k})_{k \in \geq 1}$ with commutation relations
\beq
\label{comZ} \forall k,\ell \in \mathbb{Z},\qquad [\mathcal{L}_{k},\mathcal{L}_{\ell}] = (k - \ell)\mathcal{L}_{k + \ell}\,.
\eeq
\end{theorem}
In particular, in the generic case $\theta_{-1} \neq 0$, we get a Lie algebra isomorphic to $\mathcal{L}_{-1},\mathcal{L}_{0},\mathcal{L}_{1},\ldots$, which is the Lie algebra of change of coordinates in the formal disk at $0$ in $\mathbb{C}$.

\begin{proof} I do not know an elegant proof of this result: one can check by direct computation that the six relations \eqref{fAeq}-\eqref{ABCeq}-\eqref{Deq} are satisfied. Note that $B^i_{j,k} = 0$ whenever $i + j > k +1$, and $C^i_{j,k} = 0$ whenever $i > j + k + 2$, and this implies that all sums over indices $a \geq 0$ appearing in \eqref{ABCeq} contain only finitely non-zero terms.  Let us focus on identifying the Lie algebra. We have
\beq
\label{comL} [L_i,L_j] = \sum_{k \geq 0} \frac{(2k + 1)!!}{(2i + 1)!!(2j + 1)!!} 2(j - i) \theta_{k - i - j}\,L_{k}\,.
\eeq
If we decompose
$$
\frac{z^{r_0}}{\theta} = \sum_{k \geq 0} \tau_k\,z^k\,\dd z\,,
$$
and define
$$
\mathcal{L}_{i} = - \sum_{k \geq 0} \tau_{k} \tilde{L}_{i + k - r_0},,\qquad \tilde{L}_i = \tfrac{(2i + 1)!!}{2}\,L_i\,,
$$
computation shows that the commutations \eqref{comL} imply
$$
[\mathcal{L}_i,\mathcal{L}_{j}] = (i - j)\mathcal{L}_{i + j}\,.
$$
\end{proof}

\noindent $\bullet$ \textbf{Symplectic point of view.} There is an analogy with the structure of $(A,B,C)$ in this quantum Airy structure and that described for \textsc{tqft}s: $A,B,C$ are always the result of applying the trace form to a triple product. The trace form $\psi$ of a Frobenius algebra is replaced here by the formal residue at $0$, and the product is the usual product of formal Laurent series. The vector space $W = \mathbb{C}[z^{-2},z^2]].\dd z \simeq T^*V$ admits a symplectic structure
$$
\omega(f,g) = \Res_{0} \Big(f \int g\Big)\,,
$$
and $\dd\,:\,V \rightarrow W$ is a linear Lagrangian embedding. $(\xi_k)_{k \geq 0}$ defined in \eqref{xidef} are elements of $W$, and they also span a Lagrangian subspace $V'$ of $W$, so that we get a Lagrangian splitting $W \simeq \dd V \oplus V'$, and $(\xi_k)_{k \geq 0}$ was precisely chosen such that it forms together with $(\xi_k^*)$ a symplectic basis adapted to this splitting
$$
\omega(\xi_k,\dd\xi_{\ell}^*) = \delta_{k,\ell}\,.
$$

If $(u_{i,j})_{i,j \geq 0}$ is a symmetric semi-infinite matrix, we can transform this quantum Airy structure by conjugation by $\exp\big(\tfrac{\hbar}{2}u_{a,b}\partial_{x_a}\partial_{x_b}\big)$ as described in Section~\ref{Sop}. Comparing \eqref{utran} with \eqref{ABCloop}, we see that this transformation amounts to
$$
\xi_{k} \rightarrow \xi_{k} + u_{k,a}\dd\xi_{a}^*\,.
$$
In other words, it changes the supplement subspace $V'$ in the Lagrangian splitting $W \simeq \dd V \oplus V'$.

\subsection{Witten-Kontsevich partition function}

\label{WKpart}
\noindent $\bullet$ \textbf{Moduli space of curves.} Let $\mathcal{M}_{g,n}$ be the moduli space of Riemann surfaces $\mathcal{C}$ of genus $g$ with $n$ labeled punctures $p_1,\ldots,p_n$. The automorphism group of $(\mathcal{C},p_1,\ldots,p_n)$ is finite if and only if $2g - 2 + n > 0$. Then, for $2g - 2 + n > 0$, $\mathcal{M}_{g,n}$ is a complex orbifold. Its dimension is $d_{g,n} = 3g - 3 + n$, as we will briefly justify in Section~\ref{Dcount}. $\mathcal{M}_{g,n}$ is not compact, as one can pinch cycles in $\mathcal{C}$. It can be compactified to the Deligne-Mumford moduli space $\overline{\mathcal{M}}_{g,n}$, by adding nodal surfaces whose smooth components all have finite group of automorphisms respecting the nodes and the punctures (see Figure~\ref{Fstable}). Such surfaces are called ``stable".

\begin{figure}[h!]
\begin{center}
\includegraphics[width=0.3\textwidth]{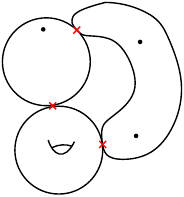}
\caption{\label{Fstable} A nodal Riemann surface in $\partial\overline{\mathcal{M}}_{2,3}$. It has three smooth components which have genus/number of punctures and nodes $(0,3)$, $(0,4)$, $(1,2)$, hence they all have finite automorphism group.}
\end{center}
\end{figure}

\vspace{0.2cm}

\noindent $\bullet$ \textbf{Intersection of $\psi$-classes.} 
For $i \in \{1,\ldots,n\}$, one can associate to each $(\mathcal{C},p_1,\ldots,p_n) \in \overline{\mathcal{M}}_{g,n}$ the complex one-dimensional space $T^*_{p_i}\mathcal{C}$. It forms over $\overline{\mathcal{M}}_{g,n}$ a holomorphic line bundle, and its first Chern class is a cohomology class in $H^2(\overline{\mathcal{M}}_{g,n},\mathbb{Q})$ denoted $\psi_i$. We also have fundamental class $[1] \in H^0(\overline{\mathcal{M}}_{g,n},\mathbb{Q})$, which is characterized by
$$
\int_{\overline{\mathcal{M}}_{g,n}} [1] = 1\,.
$$
We take the convention $\psi_i^{0} = [1]$, and if $\omega$ is a cohomology class which does not have top degree $3g - 3 + n$, we set $\int_{\overline{\mathcal{M}_{g,n}}} \omega = 0$.

\begin{theorem}
\label{WKth} Choose $\theta = z^{-2}(\dd z)^{-1}$ in the quantum Airy structure \eqref{ABCloop}. The \textsc{tr} amplitudes then compute the $\psi$-classes intersections
$$
\forall k_1,\ldots,k_n \geq 0,\qquad F_{g,n}[k_1,\ldots,k_n] = \int_{\overline{\mathcal{M}}_{g,n}} \psi_{1}^{k_1} \cup \cdots \cup \psi_{n}^{k_n}\,.
$$
We denote $Z_{{\rm WK}}$ the corresponding partition function.
\end{theorem}
This theorem means that the generating series
$$
Z = \exp\bigg\{\sum_{g \geq 0} \sum_{n \geq 1} \sum_{k_1,\ldots,k_n \geq 0} \frac{\hbar^{g - 1}}{n!}\,\bigg(\int_{\overline{\mathcal{M}}_{g,n}} \psi_{1}^{k_1} \cup \cdots \cup \psi_{n}^{k_n}\bigg) x_{k_1}\cdots x_{k_n}\bigg\}\,,
$$
satisfies the differential equations $\mathcal{L}_{i} \cdot Z = 0$ for all $i \geq -1$, where the lowest $\mathcal{L}_{i} = L_{i - 1}$ are equal to
\bea
\mathcal{L}_{-1} & = & \hbar \partial_{x_0} - \tfrac{1}{2}x_0^2 - \sum_{k \geq 0} \hbar x_{k}\partial_{x_{k + 1}}\,, \nonumber \\
\mathcal{L}_{0} & = & \hbar \partial_{x_1} - \sum_{k \geq 0} \hbar \tfrac{2k + 1}{3}\,x_{k}\partial_{x_k} - \tfrac{\hbar}{24}\,,  \nonumber \\
\mathcal{L}_{1} & = & \hbar \partial_{x_2} - \sum_{k \geq 0} \hbar \tfrac{(2k + 3)(2k + 1)}{15}x_{k}\partial_{x_{k + 1}} - \tfrac{\hbar^2}{30}\,\partial_{x_0}^2\,. \nonumber
\eea
The constraints $\mathcal{L}_{-1}\cdot Z = 0$ are $\mathcal{L}_{0}\cdot Z = 0$ can be proved in an elementary way and are called the string and dilaton equation. As $\overline{\mathcal{M}}_{0,3}$ is a point with no automorphism, we have
$$
F_{0,3}[i,j,k] = \int_{\overline{\mathcal{M}}_{0,3}} [1] = 1\,.
$$
With this as initial data, the string equation recursively determines all $\psi$-classes intersections. The result reads
$$
\int_{\overline{\mathcal{M}}_{0,n}} \psi_{1}^{k_1} \cup \cdots \cup \psi_{n}^{k_n} = \left\{\begin{array}{lll} \frac{(n - 3)!}{k_1!\cdots k_n!} & & {\rm if}\,\,\sum_{i} k_i = n - 3 \\ 0 & & {\rm otherwise} \end{array}\right.\,.
$$
It is also an elementary result in algebraic geometry that
$$ 
\int_{\overline{\mathcal{M}}_{1,1}} \psi_1 = \tfrac{1}{24}\,,
$$ 
and this is indeed the value assigned to $D^0$ in Theorem~\ref{WKth}. The equation $\mathcal{L}_{1}\cdot Z = 0$ is already non-trivial. This result is implied by the conjecture of Witten \cite{Witten} that $Z$ is a KdV tau function satisfying the string equation, proved by Kontsevich in \cite{Kontsevich}.

\newpage

\section{Cohomological field theories}

\subsection{Definition}

\noindent $\bullet$ \textbf{Axioms.} Let $\mathcal{A}$ be a Frobenius algebra over $\mathbb{C}$. A cohomological field theory (\textsc{cohft}) on $\mathcal{A}$ is a sequence of linear maps $\Omega_{g,n}\,:\,\mathcal{A}^{\otimes n} \rightarrow H^{\bullet}(\mathcal{M}_{g,n},\mathbb{\mathbb{C}})$ indexed by integers $g \geq 0$ and $n \geq 1$ such that $2g - 2 + n > 0$, satisfying the following axioms.
\begin{itemize}
\item[$\bullet$] $\Omega_{g,n}$ is symmetric under permutation of the $n$ factors of $\mathcal{A}$.
\item[$\bullet$] $\Omega_{0,3}\,:\,\mathcal{A}^{\otimes 3} \rightarrow H^{\bullet}(\overline{\mathcal{M}}_{0,3},\mathbb{C}) \cong \mathbb{C}$ represents the product in $\mathcal{A}$.
\item[$\bullet$] $\Omega$s are compatible with all natural morphisms between moduli spaces of curves, in a sense spelled out below.
\end{itemize}

\vspace{0.2cm}

\noindent $\bullet$ \textbf{Morphisms between $\overline{\mathcal{M}}_{g,n}$s.} The list of these natural morphisms is as follows. First, we can forget the last puncture, obtaining maps
$$
\pi_0\,:\,\,\overline{\mathcal{M}}_{g,n + 1} \longrightarrow \overline{\mathcal{M}}_{g,n}\,.
$$
We then require
$$
\pi_0^*\Omega_{g,n}(v_1 \otimes \cdots \otimes v_n) = \Omega_{g,n + 1}(v_1 \otimes \cdots \otimes v_n \otimes \mathbf{1})\,.
$$ 

Second, if $(p,q)$ is an ordered pair of points on a stable Riemann surface, we can glue them so as to form a nodal surface. We obtain in this way two family of morphisms
\bea
\pi_{{\rm I}} & : & \overline{\mathcal{M}}_{g - 1,n + 2} \longrightarrow \overline{\mathcal{M}}_{g,n}\,, \nonumber \\
\pi_{{\rm II}}\ &: & \overline{\mathcal{M}}_{g_1,n_1 + 1} \times \overline{\mathcal{M}}_{g_2,n_2 + 1} \longrightarrow \overline{\mathcal{M}}_{g_1 + g_2,n_1 + n_2}\,. \nonumber
\eea
 
\begin{figure}[h!]
\begin{center}
\includegraphics[width=\textwidth]{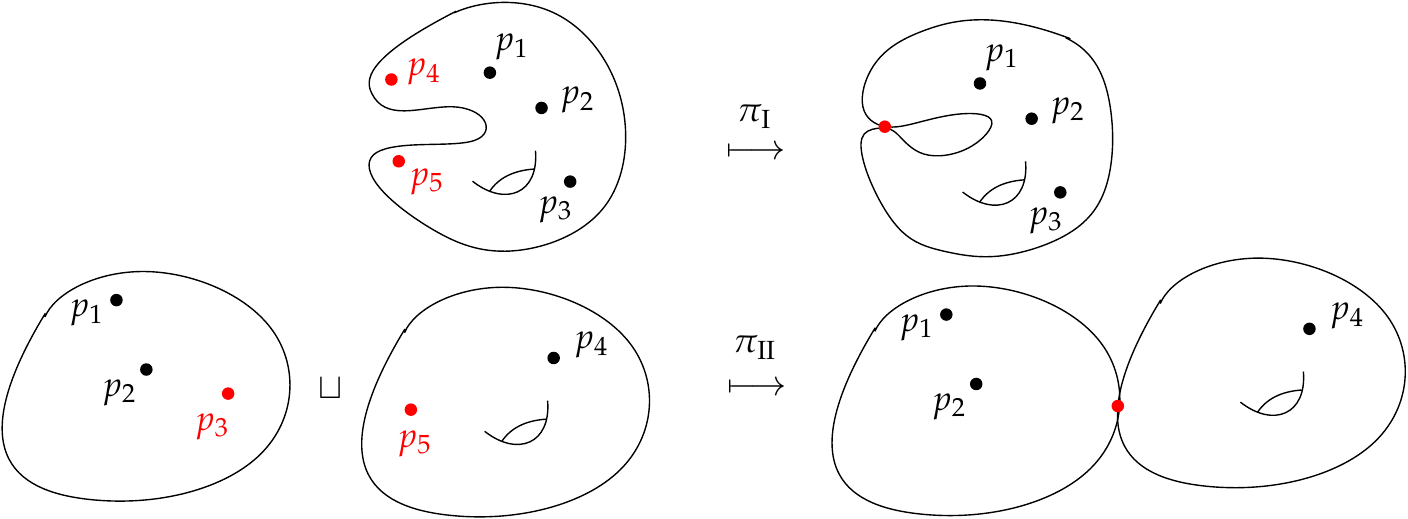}
\end{center}
\end{figure}

 The image of these morphisms lies in the boundary of the target moduli space. In fact, the collection of all morphisms $\pi_{{\rm I}}$ and $\pi_{{\rm II}}$ is the collection of inclusions of codimension $1$ strata in the boundary of $\overline{\mathcal{M}}_{g,n}$. In other words, the interior of the source moduli space in these morphisms consists of the Riemann surfaces one can obtain by degenerating in the most generic way a Riemann surface in $\mathcal{M}_{g,n}$. Let $b^{\dagger} = \sum_{i \in I} e_i \otimes e_i^*$ be the element of $\mathcal{A} \otimes \mathcal{A}$ representing the pairing in $\mathcal{A}$. We then require that
\beq
\pi_{{\rm I}}^*\Omega_{g,n}(v_1 \otimes \cdots \otimes v_n) =  \sum_{i \in I} \Omega_{g-1,n + 2}(v_1\otimes \cdots \cdots v_n \otimes e_i \otimes e_i^*)\,,
\eeq
and
\bea
&& \pi_{{\rm II}}^*\Omega_{g_1 + g_2,n_1 + n_2}(v_1 \otimes \cdots \otimes v_{n_1 + n_2}) \nonumber \\
& = & \sum_{i} \Omega_{g_1,n_1}(v_1 \otimes \cdots \otimes v_{n_1} \otimes e_i) \cup \Omega_{g_2,n_2}(v_{n_1 + 1} \otimes \cdots \otimes v_{n_1 + n_2} \otimes e_i^*)\,. \nonumber
\eea

\vspace{0.2cm}

\noindent $\bullet$ \textbf{CohFT partition function.} To simplify the exposition, let us take $(e_i)_{i \in I_0}$ an orthonormal basis of $\mathcal{A}$. If $(\Omega_{g,n})_{g,n}$ is a \textsc{cohft}, we can form the partition function $\mathfrak{Z} = \exp(\mathfrak{F})$ with
\beq
\label{compC} \mathfrak{F} = \sum_{\substack{g \geq 0 \\ n \geq 1}} \frac{\hbar^{g - 1}}{n!} \sum_{\substack{i_1,\ldots,i_n \in I_0 \\ k_1,\ldots,k_n \geq 0}} \bigg(\int_{\overline{\mathcal{M}}_{g,n}} \Omega_{g,n}(e_{i_1} \otimes \cdots \otimes e_{i_n}) \cup \psi_{1}^{k_1} \cup \cdots \cup \psi_{n}^{k_n}\bigg) \prod_{j = 1}^n x_{i_j,k_j}\,,
\eeq
which is a formal function on $V = \mathcal{A}[[z]]$. Here $x_{i,k}$ indexed by $(i,k) \in I = I_0 \times \mathbb{N}$ are the canonical linear coordinates on $V$, \textit{i.e.} corresponding to the basis $e_iz^{k}$. The components $\mathfrak{F}_{g,n}$ of \eqref{compC} are called the \textsc{cohft} amplitudes.

\vspace{0.2cm}

\noindent $\bullet$ \textbf{Degree $0$ part.} The restriction of a \textsc{cohft} to cohomological degree $0$ yields multilinear maps $\Omega_{g,n}^{{\rm deg}\,0}\,:\,\mathcal{A}^{\otimes n} \rightarrow H^{0}(\overline{\mathcal{M}}_{g,n},\mathbb{C}) \cong \mathbb{C}$. The axioms of the \textsc{cohft} imply that these maps are completely determined by $\mathcal{A}$. Indeed, we can compute $\Omega_{g,n}^{{\rm deg}\,0}$ by degenerating any stable Riemann surface to a glueing of spheres with $3$ punctures only, and use the axioms. More precisely, comparing with Section~\ref{TQFTa} we find that $\Omega_{g,n}^{{\rm deg}\,0}$ are the amplitudes of the \textsc{tqft} associated to $\mathcal{A}$. We also note that, if $\Omega$ is a \textsc{cohft}, so is $\Omega^{{\rm deg}\,0}$.

\vspace{0.2cm}

\noindent $\bullet$ \textbf{Remarks.} If a \textsc{cohft} is not purely of cohomological degree $0$, the axioms give constraints on the classes $\Omega_{g,n}$ but do not determine them. In fact, the structure of the ring $H^{\bullet}(\overline{\mathcal{M}}_{g,n})$ is rather complicated and still the theater of active research. So, there is little hope to classify \textsc{cohft}s in general. Yet, the study of \textsc{cohft}s gave recently (Pixton, Pandharipande, Zvonkine, and collaborators) new results about $H^{\bullet}(\overline{\mathcal{M}}_{g,n})$, and the classification in the case where $\mathcal{A}$ is semi-simple has been completed through the work of Givental and Teleman. This is a deep result, which uses other difficult results on the cohomology ring $H^{\bullet}(\overline{\mathcal{M}}_{g,n},\mathbb{Q})$. We will explain below how this classification works.

\vspace{0.2cm}

\noindent $\bullet$ \textbf{Prepotential and deformation of Frobenius algebras.}  We define the prepotential of a \textsc{cohft} as the formal function on $\mathcal{A}$ defined by restriction of the genus $0$ amplitude of the \textsc{cohft}
$$
\mathfrak{f}_{0} = \mathfrak{F}_{0}|_{x_{i,k} = 0\,\,k \geq 1}\,.
$$
We say that $x_{i,0}$ are the ``primary times''. Let us assume that $\mathfrak{f}_{0}$ has a non-zero radius of convergence, and take $t \in \mathcal{A}$ within the radius of convergence. For simplicity we also assume that $\mathbf{1} = e_0$  is an element of the orthonormal basis $(e_i)_{i \in I}$. One can prove from the \textsc{cohft} axioms that
\bea
\mu^{i_1}_{i_2,i_3} (t) & = & \frac{\partial^3\mathfrak{f}_{0}}{\partial t_{i_1,0}\partial t_{i_2,0} \partial t_{i_3,0}}(t) \,,\nonumber \\
b_{i_1,i_2}(t) & = & \mu^{0}_{i_2,i_3}\,, \nonumber
\eea
determine tensors $\mu(t)\,:\,\mathcal{A} \otimes \mathcal{A} \rightarrow \mathcal{A}$ and $b(t)\,:\,\mathcal{A} \otimes \mathcal{A} \rightarrow \mathbb{C}$ giving $\mathcal{A}$ a $t$-dependent structure of Frobenius algebra. We thus obtain a family of Frobenius algebras, parametrized by a neighborhood $\mathcal{N}$ of $0$ in $\mathcal{A}$, which ``derive from" the prepotential function $\mathfrak{f}_{0}$. This has been axiomatized under the name of Frobenius manifolds \cite{Dubrovin}. This notion plays an important role in enumerative geometry, especially from the perspective of mirror symmetry and integrable hierarchies.

\subsection{Two examples.}
\label{Sex}
\noindent $\bullet$ \textbf{Trivial CohFT.}  From the axioms, the fundamental class $[1] \in H^{0}(\overline{\mathcal{M}}_{g,n},\mathbb{Q})$ is automatically a \textsc{cohft} on the trivial Frobenius algebra $V = \mathbb{C}$. Its partition function is by construction $\mathfrak{Z}_{{\rm WK}}$. Therefore, it is the partition function of the quantum Airy structure given in Theorem~\ref{thAiry}. If we choose $\Delta \in \mathbb{C}^*$,
$$
\Delta^{2g - 2 + n} \cdot [1] \in H^0(\overline{\mathcal{M}}_{g,n},\mathbb{Q})
$$
is also a \textsc{cohft}, on the modified Frobenius algebra $\mathcal{A}_{\Delta} = \mathbb{C}$ defined such that the element $\mathbf{1} \in \mathbb{C}$ is orthonormal but $\mathbf{1} \cdot \mathbf{1} = \Delta \mathbf{1}$. The partition function of this \textsc{cohft} is
$$
\hat{\Delta}\mathfrak{Z}_{{\rm WK}} = \mathfrak{Z}_{{\rm WK}}(\Delta x)|_{\hbar \rightarrow \hbar \Delta^2}\,.
$$
In fact, any one-dimensional Frobenius algebra must be isomorphic to $\mathcal{A}_{\Delta}$ for some $\Delta \in \mathbb{C}^*$.

\vspace{0.2cm}

\noindent $\bullet$ \textbf{Chern character of bundles of conformal blocks.} Two-dimensional rational conformal field theory are axiomatized by the notion of modular functor -- see \textit{e.g.} \cite{BB}. We will not enter into the detail of these axiomatics, but rather state of some of its consequences. Let us fix a \textsc{cft} in this sense. There is a Frobenius algebra $\mathcal{A}$ spanned by a distinguished basis $(\varepsilon_i)_{i \in I_0}$, which describes the possible boundary conditions of the theory. For any $g,n \geq 0$ and $i_{1},\ldots,i_n \in I_0$, one can then construct the bundle of conformal blocks $\mathcal{V}_{g,n}(i_1,\ldots,i_n)$ over $\overline{\mathcal{M}}_{g,n}$, and the \textsc{cft} axioms imply that this bundle enjoys factorization properties at the boundary of the moduli space. This is the analog of the \textsc{cohft} axioms, but at the level of bundles. Then, their Chern character for $2g - 2 + n > 0$ yield a sequence of cohomology classes 
$$
\Omega_{g,n}(\varepsilon_{i_1} \otimes \cdots \otimes \varepsilon_{i_n}) = {\rm Ch}(\mathcal{V}_{g,n}(i_1,\ldots,i_n))
$$
which form a \textsc{cohft} \cite{ABO1}. Examples of \textsc{cft}s which fit in this formalism are those constructed from modular tensor categories, from representations of the quantum group $U_q(\mathfrak{sl}_2)$ at $q = \exp(\frac{2{\rm i}\pi}{k +  2})$ when $k$ is a positive integer called ``level", etc. In the latter case, the construction of the bundle of conformal blocks comes from the pioneering work of \cite{TUY}. For a general 2d rational \textsc{cft}, the existence of $\mathcal{V}_{g,n}(i_1,\ldots,i_n)$ seems to belong to folklore knowledge, and we wrote it in detail in \cite{ABO1}.

\subsection{Motivating example: quantum cohomology and Gromov-Witten theory}

\noindent $\bullet$ \textbf{Moduli space of stable maps.} The notion of \textsc{cohft} was introduced by Kontsevich and Manin \cite{KMCohFT} in an attempt to capture the properties of Gromov-Witten invariants. Let $X$ be a smooth projective variety, fix $g \geq 0$, $n \geq 1$ and $\beta \in H_{2}(X,\mathbb{Z})$. Consider the moduli space $\overline{\mathcal{M}}_{g,n}(X,\beta)$ of maps $\phi\,:\,\mathcal{C} \rightarrow X$ from a stable Riemann surface $\mathcal{C}$ of genus $g$ with $n$ labeled punctures $(p_1,\ldots,p_n)$, such that $[\phi(\mathcal{C})] = \beta$. This space is in general singular, but Behrend and Fantechi \cite{FBclass} and many others could construct a ``virtual'' fundamental cycle $[\overline{\mathcal{M}}_{g,n}(X,\beta)]_{{\rm vir}}$ over which cohomology classes can be integrated, as if they were integrated on a cycle of complex dimension
\beq
\label{dimX}d_{g,n}(X,\beta) = \dim_{\mathbb{C}} X + (3 - \dim_{\mathbb{C}} X)g + n - 3 + \int_{\beta} c_1(TX)\,.
\eeq
It is the analog of the homology cycle $[\overline{\mathcal{M}}_{g,n}]$ which is Poincar\'e dual to $[1]$ in the moduli space of curves.

\vspace{0.2cm}

\noindent $\bullet$ \textbf{Gromov-Witten classes.} This allows the definition of \emph{Gromov-Witten invariants}, as follows. We have a proper fibration $\pi\,:\,\overline{\mathcal{M}}_{g,n}(X,\beta) \rightarrow \overline{\mathcal{M}}_{g,n}$ which forgets about the map $\phi$, and $n$ morphisms ${\rm ev}_i\,:\,\overline{\mathcal{M}}_{g,n}(X,\beta) \rightarrow X$ which remember the image of $p_i$ \textit{via} $\phi$. For any $v_1,\ldots,v_n \in H^{\bullet}(X,\mathbb{C})$, we can form the cohomology class on $\overline{\mathcal{M}}_{g,n}$
$$
\Omega_{g,n}^{X,\beta}(v_1 \otimes \cdots \otimes v_n) = \pi_{*}\Big([\overline{\mathcal{M}}_{g,n}(X,\beta)]_{{\rm vir}} \cap\big( {\rm ev}_{1}^*v_1 \cup \cdots \cup {\rm ev}_{n}^*(v_n)\big)\Big)\,,
$$
which is called the Gromov-Witten theory class. Then, working in the Novikov ring $\mathbb{K}_{X}$ spanned by $t^{\beta}$ for $\beta = $ $2$-cycles which can be realized as image of a curve, we can define
\beq
\label{Novisad} \Omega_{g,n}^{X} = \sum_{\beta \in H_2(X,\mathbb{Z})} t^{\beta}\,\Omega_{g,n}^{X,\beta}\,.
\eeq

\vspace{0.2cm}

\noindent $\bullet$ \textbf{CohFTs and quantum cohomology.} $\mathcal{A}_{X} = H^{\bullet}(X,\mathbb{C})$ has the structure of a Frobenius algebra with the cup product and Poincar\'e pairing. The properties of the Gromov-Witten invariants\footnote{Some being conjectural at the time of writing of \cite{KMCohFT} as the fundamental class had not been defined yet.} guarantee that $\Omega_{g,n}^{X}$ is a \textsc{cohft} over $\mathcal{A}_{X}$ in case the sum \eqref{Novisad} is convergent -- and over $\mathcal{A}_{X}\otimes \mathbb{K}_{X}$ in general. Using the prepotential for this \textsc{cohft}, one can obtain a family of deformation of the Frobenius algebra structure on $\mathcal{A}_{X}$. The resulting ring is called ``quantum cohomology ring'' of $X$. 

\vspace{0.2cm}

\noindent $\bullet$ \textbf{Remarks.} Note that when $\dim_{\mathbb{C}} X \neq 3$, Gromov-Witten classes of high genera will vanish. So, for general varieties, the most interesting part of Gromov\-Witten theory is the genus $0$ part, and its restriction to primary times provides the quantum cohomology of $X$. When $X$ is a Calabi-Yau variety -- \textit{i.e.} $c_1(TX) = 0$ -- of complex dimension $3$, the dimension \eqref{dimX} does not depend on $\beta$ and $g$, so we have \textit{a priori} infinitely many non-vanishing Gromov-Witten invariants, and an important question is whether one can find an algorithm which can compute Gromov-Witten invariants in all genera and all degree. This question is at the heart of many developments in enumerative  and algebraic geometry in the last 25 years, relating Gromov-Witten theory to integrable systems, modular forms, etc. 

\vspace{0.2cm}

\noindent $\bullet$ \textbf{Virasoro conjecture.} Under an additional homogeneity assumption, Eguchi, Hori and Xiong  with addition of Katz \cite{EHX} proposed an explicit family $(\mathcal{L}_k)_{k \geq -1}$ of differential operators in $\mathcal{W}_{\mathcal{A}_{X}[[z]]}(\leq 2)$, which span a Lie algebra isomorphic to $\mathfrak{Witt}_{+}$, and conjecturally annihilate the \textsc{cohft} partition function. Although some cases have been proven, the general conjecture is still open, see \cite{GetzlerVir} for a review. Although they take the form \eqref{Liform}, these operators do not quite form a quantum Airy structure on $\mathcal{A}_{X}[[z]]$, as we would rather need more operators (indexed by an integer $k$ and basis elements $i$ of $\mathcal{A}_{X}$) in order to get a quantum Airy structure.

\subsection{Givental group action}
\label{Givgrou}
Let us fix a Frobenius algebra $\mathcal{A}$, and set $V = \mathcal{A}[[z]]$. Using $\mathcal{W}_{V}(\leq 2)$ seen as a quantization of $T^*V$, Givental \cite{Givquant} and Coates constructed two operations on \textsc{cohft}s, called  and $\hat{R}$, which lift at the level of cohomology classes the action of certain elements of the group $\exp(\mathcal{W}_{V}(\leq 2))$ on the partition function. The complete proof that $\hat{R}$ and $\hat{T}$ preserve the \textsc{cohft} axioms was only obtained later in \cite{FSZ}.

\vspace{0.2cm}
 
\noindent $\bullet$ \textbf{$\hat{T}$ operation.}  Let $T \in z^2\mathcal{A}[[z]] \subset V$, which we decompose as $T(z) = \sum_{k \geq 2} T_kz^{k}$. If $\Omega$ is a \textsc{cohft} on $\mathcal{A}$, we put
\bea
&& (\hat{T}\Omega_{g,n})(v_1 \otimes \cdots \otimes v_n) \nonumber \\
& = & \sum_{\substack{m \geq 0 \\ k_1,\ldots,k_m \geq 0}} \frac{1}{m!} (\pi_{m})_*\bigg( \Omega_{g,n+ m}\Big(\bigotimes_{i = 1}^n v_i \otimes \bigotimes_{j = 1}^m T_{k_j}\Big) \cup \psi_{1}^{k_1} \cup \cdots \cup \psi_{n}^{k_n}\bigg)\,, \nonumber
\eea
where we have introduced the morphism $\pi_m\,:\,\overline{\mathcal{M}}_{g,n + m} \rightarrow \overline{\mathcal{M}}_{g,n}$  forgetting the first $m$ punctures. The pushforward $(\pi_m)_*(\omega)$ here means that we take the homology class dual to $\omega$, push it forward, and then apply again Poincar\'e duality to get a cohomology class. The effect on the \textsc{cohft} partition function is a translation, as described in Section~\ref{Sop}.
$$ 
\hat{T}\mathfrak{Z}(x) = \exp\big(-\frac{1}{\hbar}\tfrac{\partial_{t_a}\mathfrak{F}_{0}(T) x_a}{\hbar} - \tfrac{\partial_{t_at_b}\mathfrak{F}_{0}(T)x_ax_b}{2\hbar}\big)\,\mathfrak{Z}(x + T)\,.
$$ 

\vspace{0.2cm}
 
\noindent $\bullet$ \textbf{$\hat{R}$-operation.} Let $R \in {\rm End}(\mathcal{A}[[z]]) \cong {\rm End}(\mathcal{A})[[z]]$ such that
\beq
\label{Rcon}R = {\rm Id}_{\mathcal{A}} + O(z)\quad {\rm and}\quad R(z)R^{\dagger}(-z) = {\rm Id}_{\mathcal{A}}\,,
\eeq
where $\dagger$ is the notion of adjoint induced by the pairing on $\mathcal{A}$. We introduce the element $B \in \mathcal{A}\otimes \mathcal{A}[[z_1,z_2]]$ through
$$
B(z_1,z_2) = \frac{b^{\dagger} - R(z_1)\otimes R(z_2) b^{\dagger}}{z_1 + z_2}\,.
$$
The conditions \eqref{Rcon} ensure that it is a well-defined formal power series in $z_1$ and $z_2$. In general, if $\omega$ is a positive degree element of a cohomology ring $\mathcal{H}$, we can evaluate an $\mathcal{A}$-valued formal power series to an element of $\mathcal{A} \otimes \mathcal{H}$ by replacing $z^{k}$ with the $k$-th power (for the cup product) of $\omega$. The series will terminate as $\omega$ is nilpotent in $\mathcal{H}$.

\begin{definition}
For $2g - 2 + n > 0$, $\mathcal{G}_{g,n}$ is the set of graphs $G$ with the following properties.
\begin{itemize}
\item[$\bullet$] $G$ has $n$ labeled leaves. We denote $L(G)$ the set of leaves.
\item[$\bullet$] Each vertex $v$ carries a genus $h(v) \geq 0$ satisfying $2h(v) - 2  + n(v) > 0$ where $n(v)$ is the number of edges/leaves incident to $v$. We denote $V(G)$ the set of vertices, and $E(G)$ the set of edges (excluding leaves).
\item[$\bullet$] $b_1(\Gamma) + \sum_{v \in V(G)} h(v) = g$.
\end{itemize}
For any $G \in \mathcal{G}_{g,n}$, we can form the moduli space
\beq
\label{modG} \overline{\mathcal{M}}_{G} = \prod_{v \in V(G)} \overline{\mathcal{M}}_{h(v),n(v)}\,,
\eeq
which has a natural inclusion $\pi_{G}\,:\,\overline{\mathcal{M}}_{G} \rightarrow \overline{\mathcal{M}}_{g,n}$.
\end{definition}

\begin{figure}[h!]
\begin{center}
\includegraphics[width=0.15\textwidth]{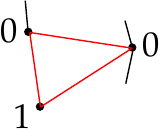}
\caption{\label{FStableG} The graph in $\mathcal{G}_{2,3}$ corresponding to the stable surface of Figure~\ref{Fstable}.}
\end{center}
\end{figure}
In \eqref{modG}, the half-edges/leaves naturally index the punctures in surfaces belonging to the target moduli space. We denote $\overline{E}(G)$ the set of half-edges/leaves of $G$. If $e$ is an edge incident at a vertex $v$, we denote $e(v)$ the corresponding half-edge incident to $v$. $(\mathcal{M}_{G})_{G \in \mathcal{G}_{g,n}}$ is actually the family of boundary strata of $\overline{\mathcal{M}}_{g,n}$. Vertices of the graph represent smooth components of a possibly nodal surface, the edges record how punctures of the smooth components should be glued together so as to form nodes, and the leaves correspond to the punctures which remain in the nodal surface.

If $\Omega$ is a \textsc{cohft} on $\mathcal{A}$, the $\hat{R}$-operation is defined by
$$
(\hat{R}\Omega)_{g,n} = \sum_{G \in \mathcal{G}_{g,n}} (\pi_{G})_* \circ \Phi \bigg(\bigcup_{\substack{\{v,v'\} \in E(G)}} B(\psi_{e(v)},\psi_{e(v')}) \cup \bigcup_{\ell \in L(G)} R(\psi_{v})\bigg)\,.
$$
Here, the result of the product between the round brackets is an element of
$$
(\mathcal{A}^*)^{\otimes \overline{E}(G)} \otimes \bigotimes_{v \in V(G)} H^{\bullet}(\overline{\mathcal{M}}_{h(v),n(v)},\mathbb{C})\,,
$$
where we used the pairing to replace copies of $\mathcal{A}$ by $\mathcal{A}^*$. The multilinear map $\Phi$ applies the pairing on the glued half-edges so as to obtain an element of
$$
(\mathcal{A^*})^{\otimes L(G)} \otimes \bigotimes_{v \in V(G)} H^{\bullet}(\overline{\mathcal{M}}_{h(v),n(v)},\mathbb{C})\,.
$$
Finally, pushing forward \textit{via} the morphism $\pi_{G}$ produces an element of
$$
(\mathcal{A^*})^{\otimes L(G)} \otimes \bigotimes_{v \in V(G)} H^{\bullet}(\overline{\mathcal{M}}_{g,n},\mathbb{C})\,.
$$

This set of graphs is in fact tailored such that the action of this transformation on the partition function is a composition of conjugation by exponential of a pure quadratic differential operator (responsible for the weight of the edges) followed by a change of basis (responsible for the weights of the leaves). Let us choose an orthonormal basis $(e_i)_{i \in I_0}$ of $\mathcal{A}$, and denote $(u_{a,b})_{a,b \in I_0 \times \mathbb{N}}$ the coefficients of $B$ in the basis $e_{i}z^k$ indexed by $a = (i,k) \in I_0 \times \mathbb{N}$. Likewise we denote $r_{a,b}$ the coefficients of $R$. We have
$$
\hat{R}\mathfrak{Z} = \exp\big(r_{a,b}x_a\partial_{x_b}\big)\exp\big(\tfrac{\hbar}{2}u_{a,b}\partial_{x_a}\partial_{x_b}\big)\mathfrak{Z}\,.
$$
In this sense, $\mathcal{G}_{g,n}$ are the Feynman graphs whose sum computes formally the right-hand side in terms of the original partition function. It is very different in nature from the set of \textsc{tr} graphs $\mathbb{G}_{g,n}$ introduced in Section~\ref{Sgraph} -- the latter are not Feynman graph, due to the presence of the spanning tree and the non-local constraint on edges which are not in the spanning tree.

\subsection{Semi-simple cohomological field theories}

\noindent $\bullet$ \textbf{Classification.} A \textsc{cohft} on $\mathcal{A}$ is said semi-simple if $\mathcal{A}$ itself is isomorphic to a direct sum of one-dimensional Frobenius algebras. C.~Teleman proved the remarkable result \cite{Teleman}

\begin{theorem}
\label{Teleth} If $\Omega$ is a semi-simple \textsc{cohft}, there exists $R \in {\rm End}\,\mathcal{A}[[z]]$ such that, for $T(z) = z({\rm Id}_{\mathcal{A}} - R(z))(\mathbf{1})$, we have $(\hat{T}\hat{R})\Omega_{g,n} = \Omega^{{\rm deg}\,0}_{g,n}$ for any $2g - 2 + n > 0$.
\end{theorem}
Under extra homogeneity assumptions on the Frobenius manifold attached to the \textsc{cohft}, one can construct a canonical $R$ from the solution of a Riemann-Hilbert problem on $\mathbb{P}^1$ attached to the semi-simple Frobenius manifold \cite{Dubrovin,Giventals}.

\vspace{0.2cm}

\noindent $\bullet$ \textbf{Computation \textit{via} topological recursion.} Here we harvest the fruits of the previous discussions. Let $\Omega$ be a semi-simple \textsc{cohft} on $\mathcal{A}$, for which we assume that $R$ -- and thus $T$ -- allowing reconstruction by Theorem~\ref{Teleth} is known. By semi-simplicity, there exists a basis $(e_{\alpha})_{\alpha \in I_0}$  which we can choose orthonormal, and such that the product is of the form
$$
\forall \alpha,\beta \in I_0,\qquad e_{\alpha} \cdot e_{\beta} = \Delta_{\alpha}\,e_{\alpha}\,\delta_{\alpha,\beta}
$$
for some $\Delta_{\alpha} \in \mathbb{C}^*$. We denote $(e_{\alpha}^*)_{\alpha \in I_0}$ its dual basis.
Then
$$
\Omega_{g,n} = \hat{T}\hat{R}\bigg(\bigoplus_{\alpha} \Delta_{\alpha}^{2g - 2 + n} (e_{\alpha}^*)^{\otimes n} [1]\bigg)\,,
$$
which at the level of partition functions gives
$$ 
\mathfrak{Z} = \hat{T}\hat{R} \prod_{\alpha \in I_0} (\hat{\Delta}_{\alpha}\mathfrak{Z}_{{\rm WK}})\big((x_{\alpha,k})_{k \geq 0}\big)\,.
$$
We already know that $\hat{\Delta}Z_{{\rm WK}}$ is the partition function of the quantum Airy structure of Theorem~\ref{thAiry} with $\theta = \Delta^{-1}z^{-2}(\dd z)^{-1}$ taking into account the rescaling. The direct sum over $\alpha \in I_0$ of those quantum Airy structures is still a quantum Airy structure, which we denote $L^{{\rm deg}\,0}$, and its partition function is the \textsc{cohft} partition function of $\Omega^{{\rm deg}\,0}$. As the operations of Givental act on the \textsc{cohft} partition functions like the operations of quantum Airy structures do on the \textsc{tr} partition function, we deduce that the \textsc{tr} partition function of
$$ 
L =  U L^{{\rm deg}\,0} U^{-1},\qquad U = \hat{T}\hat{R} L^{{\rm deg}\,0}
$$
is the \textsc{cohft} partition function of $\Omega$. This result was first proved in \cite{DBOSS}, in the language of the original Eynard-Orantin topological recursion  -- see \textit{e.g.} \cite[Section 4.2]{ABO1} for an explicit description in this language of the initial data in terms of $R$.

\newpage
\lecture[1h]{11}{05}{2017}

\section{Volumes of the moduli space of curves}

\subsection{Teichm\"uller space of bordered surfaces}
\label{Dcount}

We review some relevant facts of hyperbolic geometry, whose proof can be found \textit{e.g.} in \cite{Buser}.

\vspace{0.2cm}

\noindent $\bullet$ \textbf{Definitions.} Let $\Sigma_{g,n}$ be a topological, compact, oriented surface of genus $g$ with $n$ boundaries $\beta_1,\ldots,\beta_n$. We assume $2g - 2 + n > 0$ and fix $L_1,\ldots,L_n \in \mathbb{R}_{> 0}$. The Teichm\"uller space of bordered surfaces is
$$ 
\mathcal{T}_{g,n}(L_1,\ldots,L_n) = \big\{{\rm Riemannian}\,\,{\rm metric}\,\,\sigma\,\,{\rm on}\,\,\Sigma_{g,n}\,\,:\,\,\ell_{\sigma}(\beta_i) = L_i\big\}\big/\sim\,,
$$ 
where $\sigma \sim \sigma'$ if they are related by a conformal transformation, \textit{i.e.} there exists a continuous function $f\,:\,\Sigma_{g,n} \rightarrow \mathbb{R}_{> 0}$ such that $\sigma = \exp(f)\sigma'$. Since $2g - 2 + n > 0$, in each conformal class there exists a unique representative which is a hyperbolic (\textit{i.e.} with constant curvature $-1$) metric such that the boundaries are geodesic. We always identify points in the Teichm\"uller space with their hyperbolic representative. If $\Sigma'_{g,n}$ is a surface with labeled boundaries and same topology, and $\varphi\,:\,\Sigma'_{g,n} \rightarrow \Sigma_{g,n}$ is a diffeomorphism, the Teichm\"uller space for $\Sigma'_{g,n}$ is canonically isomorphic to the one for $\Sigma_{g,n}$.

The mapping class group $\Gamma_{g,n}$ is the group of diffeomorphisms of $\Sigma_{g,n}$ to itself which preserve the labeling of the boundaries, modulo those which are isotopic to identity among such diffeomorphisms. The moduli space of bordered surfaces is by definition
$$
\mathcal{M}_{g,n}(L_1,\ldots,L_n) = \mathcal{T}_{g,n}(L_1,\ldots,L_n)/\Gamma_{g,n}\,.
$$
\vspace{0.2cm}

\noindent $\bullet$ \textbf{Hyperbolic pair of pants.} We have $\mathcal{T}_{0,3}(L_1,L_2,L_3) = \{{\rm pt\}}$. This means that if $L_1,L_2,L_3 > 0$ are fixed, there exists a unique hyperbolic pair of pants with labeled geodesic boundaries of lengths $L_1,L_2,L_3$. It is obtained by glueing two right-angled hyperbolic hexagons, related by an isometric involution (Figure~\ref{Fig5} later).

\vspace{0.2cm}

\noindent $\bullet$ \textbf{Collars.} We call collar of waist $\ell$ the cylinder $C(\ell) = [-w(\ell),w(\ell)] \times \mathbb{S}^1$ where
$$
w(\ell) = {\rm arcsinh}\Big(\frac{1}{{\rm sinh}(\frac{\ell}{2})}\Big)\,.
$$
We call standard collar $C(\ell)$ equipped with the metric
\beq
\label{mccollar} \dd r^2 + \ell^2{\rm cosh}^2(r) \dd t^2\,,
\eeq
where $(r,e^{2{\rm i}\pi t})$ are coordinates on $[-w(\ell),w(\ell)] \times \mathbb{S}^1$. This metric is hyperbolic, and the level curves of $r$ are geodesics.  If $\theta \in \mathbb{R}$, we define the diffeomorphism $\Phi_{\theta}\,:\,C(\ell) \rightarrow C(\ell)$ by
$$
\Phi_{\theta}\big(r,t\big) = \big(r,e^{2{\rm i}\pi(t + \frac{(r + w)\theta}{\ell})}\big)\,.
$$

Let $\sigma$ be a hyperbolic metric on $\Sigma_{g,n}$. If $\gamma$ is a simple closed curve in the interior of $\Sigma$, there exists a unique representative $\gamma'$ in its free homotopy class, which is a geodesic of minimal length. Besides, $\gamma'$ has a neighborhood $C(\ell_{\sigma}(\gamma'))$ which is isometric to the standard collar $C(\ell_{\sigma}(\gamma'))$ such that the image of $\gamma'$ is sent to the circle of equation $r = 0$. It is called ``the collar neighborhood'' of $\gamma'$.

If $\beta$ is a boundary component of $\Sigma_{g,n}$ equipped with a hyperbolic metric $\sigma$, it admits a neighborhood which is isometric to $C^+(w(\ell_{\sigma})) = [0,w(\ell_{\sigma})] \times \mathbb{S}^1$ equipped again with the metric \eqref{mccollar}. This can be deduced from the previous situation by doubling $\Sigma_{g,n}$ along the boundaries.

\vspace{0.2cm}

\noindent $\bullet$ \textbf{Fenchel-Nielsen coordinates.} Fix $ (\alpha_i)_{i = 1}^q$ a family of simple closed curves such that cutting along these curves gives a pair of pants decomposition of $\Sigma_{g,n}$, which we denote $\mathbf{P} = (P_j)_{j = 1}^{p}$. The image of $\alpha_i$ in the pairs of pants consists in two curves $\alpha_i^{+},\alpha_i^{-}$ together with a diffeomorphism $\phi_i\,:\,\alpha_i^{-} \rightarrow \alpha_i^+$ to identify them pointwise. As the Euler characteristic of $\Sigma_{g,n}$ is $2 - 2g - n$, and the Euler characteristic of a pair of pants is $-1$, we must have $p = 2g - 2 + n$ pairs of pants. The number of boundaries of these pairs of pants is $3p = n + 2q$, therefore $q = 3g - 3 + n$. We can define a map
$$
\begin{array}{ccc} (\mathbb{R}_{> 0} \times \mathbb{R})^{3g - 3 + n} & \longrightarrow & \mathcal{T}_{g,n}(L_1,\ldots,L_n) \\
(\ell_i,\theta_i)_{i = 1}^{3g - 3 + n} & \longmapsto & \sigma_{\mathbf{\ell},\theta}
\end{array}\,.
$$
$\sigma_{\ell,\theta}$ is obtained as follows. We first define a metric $\sigma_{\mathbf{\ell},\mathbf{0}}$ obtained by glueing the unique hyperbolic metrics on $P_j$ making $\partial P_j$ geodesic, and such that
$$
\ell_{\sigma_{\mathbf{\ell},\mathbf{0}}}(\alpha_i) = \ell_i\qquad  {\rm and} \qquad \ell_{\sigma_{\mathbf{\ell},\mathbf{0}}}(\beta_i) = L_i\,.
$$
Then, for each $i \in \{1,\ldots,3g - 3 + n\}$, in the collar neighborhoods $C(w_{\ell}(\alpha_i))$ of $\alpha_i$ for $\sigma_{\mathbf{\ell},\mathbf{0}}$, we replace $\sigma_{\mathbf{\ell},\mathbf{0}}$ by its image through the twist diffeomorphisms $\Phi_{\theta_i}$, and reuniformize so that we get a hyperbolic metric.

This map is in fact a diffeomorphism, and $(\mathbf{\ell},\theta)$ are called the Fenchel-Nielsen coordinates.  They depend on the pair of pants decomposition. Therefore, the real dimension of $\mathcal{T}_{g,n}(L_1,\ldots,L_n)$ is $6g - 6 + 2n$, and $\mathcal{M}_{g,n}(L_1,\ldots,L_n)$ has same dimension. This is still true when one extends the previous definitions to allow $L_i = 0$. In that case $\mathcal{M}_{g,n}(0,\ldots,0)$ is the moduli space of curves $\mathcal{M}_{g,n}$ introduced Section~\ref{WKpart}.

\vspace{0.2cm}

\noindent $\bullet$ \textbf{Action of the mapping class group.} $\Gamma_{g,n}$ acts transitively on the set of pairs of pants decompositions of $\Sigma_{g,n}$. Let ${\rm Stab}(\mathbf{P})$ the subgroup of $\Gamma_{g,n}$ which preserves the pair of pants decomposition $\mathbf{P}$. It is generated by the Dehn twists $\hat{\delta}_{i}$ along the curves $\alpha_i$ for $i \in \{1,\ldots,3g - 3 + n\}$. If $(g,n) \neq (1,1)$, $\hat{\delta}_{i}$ acts by $\theta_i \rightarrow \theta_i +  \ell_i$, leaving the other coordinates invariant. For the torus with one boundary, there is a single pair of pants in which two boundaries are glued together to form the curve $\alpha_1$, and the half-Dehn twist acting as $\theta_1 \rightarrow \theta_1 + \tfrac{\ell_1}{2}$ generates ${\rm Stab}(\mathbf{P})$.

\vspace{0.2cm}

\noindent $\bullet$ \textbf{Weil-Petersson form.} Given a pair of pants decomposition $\mathbf{P}$, one can define the Weil-Petersson symplectic form
\beq
\label{WPform}\omega_{{\rm WP}} = \sum_{i = 1}^{3g - 3 + n} \dd\ell_i \wedge \dd\theta_i\,.
\eeq
It is obviously invariant under ${\rm Stab}(\mathbf{P})$. By a theorem of Wolpert \cite{Wolpert1}, it is in fact invariant under $\Gamma_{g,n}$, therefore it defines canonically a symplectic form on $\mathcal{T}_{g,n}(L_1,\ldots,L_n)$, which descends to the moduli space $\mathcal{M}_{g,n}(L_1,\ldots,L_n)$.  The Weil-Petersson volume form is
\beq
\label{WPform2} \nu_{{\rm WP}} = \omega_{{\rm WP}}^{\wedge (3g - 3 + n)}\,.
\eeq
Basic estimates show that $\mathcal{M}_{g,n}(L_1,\ldots,L_n)$ has a finite volume for the Weil-Petersson form, which we denote $V_{g,n}(L_1,\ldots,L_n)$.

\subsection{Mirzakhani's recursion}

\noindent $\bullet$ \textbf{Main formulas.} By convention, we set $V_{0,1}(L) = V_{0,2}(L_1,L_2) = 0$. As the moduli space of bordered pairs of pants is always a point, we have
\beq
\label{V03} V_{0,3}(L_1,L_2,L_3) = 1\,.
\eeq
Early computations of Weil-Petersson volumes for low Euler characteristic can be found in \cite{Penner}. In particular, one can describe explicitly the moduli space of tori with one boundary, and one finds
\beq
\label{V11} V_{1,1}(L) = \zeta(2) + \tfrac{L^2}{24} = \tfrac{\pi^2}{6} + \tfrac{L^2}{24}\,.
\eeq
This serves as initial data for the recursion found by Mirzakhani \cite{Mirza1}. We introduce auxiliary functions
\bea
f(z) & = & -2 \ln(1 + e^{-z/2})\,, \nonumber \\
\mathcal{B}(L_1,L_2,L_3) & = & \tfrac{L_3}{2L_1}\big(f(L_3 + L_2 + L_1) - f(L_3 + L_2 - L_1) \nonumber \\
&& \quad\quad + f(L_3 - L_2 + L_1) - f(L_3 - L_2 - L_1)\big)\,, \nonumber \\
\mathcal{C}(L_1,L_2,L_3) & = & \tfrac{L_2L_3}{L_1}\big(f(L_3 + L_2 + L_1) - f(L_3 + L_2 - L_1)\big)\,. \nonumber \\
\label{Ccal}
\eea

\begin{theorem}
\label{thMirza} For $2g - 2 + n \geq 2$, we have
\bea
&& V_{g,n}(L_1,\ldots,L_n) \nonumber \\
& = & \sum_{m = 2}^{n} \int_{\mathbb{R}_{> 0}} \dd \ell\,\mathcal{B}(L_1,L_m,\ell) V_{g,n - 1}(\ell,L_2,\ldots,\widehat{L_m},\ldots,L_n) \nonumber \\
& & + \tfrac{1}{2} \int_{\mathbb{R}_{> 0}^2} \dd \ell \dd \ell'\,\mathcal{C}(L_1,\ell,\ell')\bigg(V_{g - 1,n + 1}(\ell,\ell',L_2,\ldots,L_n) + \nonumber \\
&& \phantom{ \tfrac{1}{2} \int_{\mathbb{R}_{> 0}^2} \dd \ell \dd \ell'\,\mathcal{C}(L_1,\ell,\ell')} + \sum_{\substack{g_1 + g_2 = g \\ J_1 \dot{\cup} J_2 = \{L_2,\ldots,L_n\}}} \!\!\!\!\!\!\! V_{g_1,1+|J_1|}(\ell,J_1)V_{g_2,1+|J_2|}(\ell',J_2)\bigg)\,. \nonumber
\eea
\end{theorem}

\vspace{0.2cm}

\noindent $\bullet$ \textbf{Correspondence with the topological recursion.} This recursion has the same structure as \textsc{tr} (remember Section~\ref{Section1}). The correspondence was first established by Eynard and Orantin \cite{EOwp} using the language of \cite{EOFg} with residues. Here, we will describe it in the language of quantum Airy structures. We put $V^* = \mathbb{C}[L^2]$, so that linear coordinates on $V$ are $x_i = L^{2i}$ indexed by $i \in I = \mathbb{N}$. Let us define
$$
\theta_k = \zeta(2k + 2)(2^{2k + 3} - 4),\qquad k \geq -1\,.
$$
A generating series for this sequence of numbers is
$$
\theta(z) = \sum_{k \geq -1} \theta_i\,z^k (\dd z)^{-1} = \frac{\pi z}{\sin(2\pi z)\dd z}\,.
$$
\begin{proposition}
\label{VQ}The following data
\bea
A^i_{j,k} & = & \delta_{i,j,k,0}\theta_{-1}\,, \nonumber \\
B^i_{j,k} & = & \frac{(2k + 1)!}{(2i + 1)!(2j + 1)!}\,(2j + 1)\,\theta_{k - j - i}\,, \nonumber \\
C^i_{j,k} & = & \frac{(2j + 1)!(2k + 1)!}{(2i + 1)!}\,\theta_{k + j + 1 - i}\,, \nonumber \\
D^i & = & \tfrac{\pi^2}{6}\,\delta_{i,0} + \tfrac{1}{24}\,\delta_{i,1}\,, \nonumber
\eea
is quantum Airy structure, and its amplitudes are such that
\beq
\label{VTR} V_{g,n}(L_1,\ldots,L_n) = \sum_{k_1,\ldots,k_n \geq 0} F_{g,n}[k_1,\ldots,k_n] L_{1}^{2k_1}\cdots L_n^{2k_n}\,.
\eeq
\end{proposition}
In fact, it is a special case of the quantum Airy structure of Theorem~\ref{thAiry} provided we use the isomorphism
$$
\begin{array}{ccc} \mathbb{C}[L^2] & \longrightarrow & z.\mathbb{C}[[z^2]] \\
L^{2i} & \longmapsto & 2^ii!\,e_i \end{array}\,,
$$
where $e_i = \tfrac{z^{2i + 1}}{(2i + 1)!!}$ is the reference basis in the target vector space.

\vspace{0.2cm}

\noindent $\bullet$ \noindent \textbf{Proof of Proposition~\ref{VQ}.} To apply~\ref{thMirza}, we need to compute the auxiliary integral
$$
H_{k}(x) = \int_{\mathbb{R}_{> 0}} \dd \ell\,\ell^{2k + 1}\big(f(\ell + x) - f(\ell - x)\big)\,.
$$
\begin{lemma}
For any integer $k \geq 0$, $H_k$ is an odd polynomial of degree $2k + 3$
$$
H_k(x) = (2k + 1)! \sum_{i = 0}^{k + 1} \theta_{i - 1}\,\frac{x^{2k + 3 - 2i}}{(2k + 3 - 2i)!} \,.
$$
\end{lemma}
\begin{proof} We first remark that $f(z) - f(-z) = z$. Integrating twice by parts we compute
$$
H_{k}''(x) = 2k(2k + 1)H_{k - 1}(x) + \delta_{k,0}x\,,
$$
which is valid for any $k \geq 0$ with the convention $H_{-1}(x) = 0$. Therefore, for any $i \leq k$
$$
H_{k}^{(2i + 1)}(0) = \frac{(2k + 1)!}{(2k + 1 - 2i)!}\,H_{k - i}'(0)\,,
$$
and we deduce that $H_k$ is an odd polynomial of degree $2k + 3$ with leading coefficient $(2k + 1)!x^{2k + 3}$, given by the Taylor sum
\beq
\label{HKsum} H_k(x) = (2k + 1)!\,x^{2k + 3} + \sum_{i = 0}^{k} \frac{(2k + 1)!}{(2k + 1 - 2i)!}\,H_{k - i}'(0)\,\frac{x^{2i + 1}}{(2i + 1)!}\,.
\eeq
On the other hand, we can compute directly for any $k \geq 0$
\bea 
H_{k}'(0) & = & 2\int_{\mathbb{R}_{> 0}} \frac{\dd \ell\,\ell^{2k + 1}}{1 + \exp(\frac{\ell}{2})} \nonumber \\
& = & \sum_{m \geq 1} (-1)^{m + 1} \int_{\mathbb{R}_{> 0}} \dd \ell\,\ell^{2k + 1}\,\exp(-\tfrac{m\ell}{2}) \nonumber \\
& = & (2k + 1)! \sum_{m \geq 1} (-1)^{m + 1}\,2\,(\tfrac{2}{m})^{2k + 2} \nonumber  \\
& = & (2k + 1)!\,2^{2k + 3}\big(-\zeta_{{\rm even}}(2k + 2) + \zeta_{{\rm odd}}(2k + 2)\big)\,, \nonumber
\eea 
where
$$
\zeta_{{\rm even}}(s) = \sum_{n \geq 1} (2n)^{-s},\qquad \zeta_{{\rm odd}}(s) = \sum_{n \geq 0} (2n + 1)^{-s}
$$ 
We have the relations $\zeta_{{\rm even}}(s) = 2^{-s}\zeta(s)$ and $\zeta_{{\rm odd}}(s) = \zeta(s) - \zeta_{{\rm even}}(s)$, therefore
$$
H_{k}'(0) = (2k + 1)!(2^{2k + 3} - 4)\zeta(2k + 2) = \theta_{k}\,.
$$
Besides, we notice that the first term in \eqref{HKsum} is $(2k + 1)! = (2k + 1)!\theta_{-1}$, therefore it can be included in the sum as the term $i = k + 2$.
\end{proof}
We then have
\beq
\int_{\mathbb{R}_{> 0}} \dd\ell\,\mathcal{B}(L_1,L_2,\ell)\,\ell^{2k} = \tfrac{1}{2L_1}\big(H_{k}(L_1 + L_2) - H_{k}(L_1 - L_2)\big)\,, \nonumber
\eeq
\bea
\int_{\mathbb{R}_{> 0}^2} \dd \ell \dd \ell' \,\mathcal{C}(L,\ell,\ell')\,\ell^{2j} (\ell')^{2k} & = & \frac{1}{L_1} \int \dd\ell\dd\ell'\big(f(\ell + \ell' + L_1) - f(\ell + \ell' - L_1)\big) \nonumber \\
& = & \frac{(2j + 1)!(2k + 1)!}{(2j + 2k + 3)!}\,\frac{H_{j + k + 1}(L)}{L}\,. \nonumber
\eea
In the last line, we used the change of variable $(\ell,\ell') \rightarrow (\ell + \ell',\ell)$, and the prefactor of $H_{j + k + 1}$ in the second line arises from the integral over the second variable which gives the Euler beta integral. We deduce recursively from Theorem~\ref{thMirza} that $V_{g,n}(L_1,\ldots,L_n) \in \mathbb{C}[L_1^2,\ldots,L_n^2]$, and comparing with the definition of the \textsc{tr} amplitudes, that \eqref{VTR} holds for $2g - 2 + n > 0$ with
\bea
B^i_{j,k} & = & [L_1^{2i} L_2^{2j}]\,\,\tfrac{1}{2L_1}\big(H_{k}(L_1 + L_2) - H_{k}(L_1 - L_2)\big)\,, \nonumber \\
C^i_{j,k} & = & [L^{2i}]\,\,\frac{(2j + 1)!(2k + 1)!}{(2j + 2k + 3)!}\,\frac{H_{j + k + 1}(L_1)}{L_1}\,. \nonumber
\eea
$(A,D)$ are identified by comparison with \eqref{V03}-\eqref{V11} so that \eqref{VTR} holds for $(g,n) = (0,3)$ and $(1,1)$. The fact it is a quantum Airy structure follows by comparison with the quantum Airy structure in Theorem~\ref{thAiry}.

\subsection{Sketch of the proof of Theorem~\ref{thMirza}}

We now explain some steps of the proof of the main recursion, following Mirzakhani. We assume $2g - 2 + n \geq 2$, and fix $\sigma \in \mathcal{T}_{g,n}(L_1,\ldots,L_n)$.

\vspace{0.2cm}

\noindent $\bullet$ \textbf{Partition of $\beta_1$.} We are going to partition $\beta_1$ into pieces depending on the behavior of the geodesic $\gamma_{x}$ issueing from $x \in \beta_1$ orthogonally to $\beta_1$, and write
$$
L_1 = \sum_{{\rm pieces}} \ell_{\sigma}({\rm piece})\,.
$$

\begin{definition}
We say that $\gamma_{x}$ is an orthogeodesic if it is a simple curve which is, at any point of intersection with the boundary, orthogonal to this boundary. We denote $E$ the set of $x \in \beta_1$ for which $\gamma_{x}$ is an orthogeodesic.
\end{definition}

\begin{lemma}
\label{0mes} $E$ has zero Hausdorff dimension in $\beta_1$.
\end{lemma}
\begin{proof} By a result of Birman and Series \cite{BirmanSeries}, the union of all simple closed geodesics on a hyperbolic surface is a set of Hausdorff dimension $1$. \textit{A fortiori}, the intersection $E'$ of the half-collar neighborhood of $\beta_1$ with the union of all $\gamma_{x}$ for $x \in E$, has Hausdorff dimension $1$. Given the behavior of geodesics in the half-collar geometry, this implies that $E' \cap \beta_1 = E$ has zero Hausdorff dimension in $\beta_1$.
\end{proof}

In particular, $E$ has no interior points. Mirzakhani then established a classification of the points in $E$ depending on the behavior of the orthogeodesics, which we will not justify here.
\begin{definition}
If $\gamma,\gamma'$ are two simple curves which have no intersection, we say that $\gamma$ spirals around $\gamma'$ if any $x \in \gamma'$ is a limit point of $\gamma$.
\end{definition}
\begin{lemma}
\label{clasgeod} A point $x \in E$ is
\begin{itemize}
\item[$\bullet$] isolated if and only if $\gamma_{x}$ meets a boundary of $\Sigma_{g,n}$ (in that case it meets it orthogonally) or spirals around a boundary of $\Sigma_{g,n}$.
\item[$\bullet$] is a limit point if and only if $\gamma_{x}$ spirals around a simple closed curve in the interior of $\Sigma_{g,n}$.
\end{itemize}
\end{lemma}
Let us denote $\beta_1'$ the set $\beta_1 \setminus E$ to which one adds back the isolated points of $E$. It must be a countable union of open segments $(I_j)_{j \in J}$. And, according to Lemma~\ref{0mes}, we have
\beq
\label{tegh}L_1 = \sum_{j \in J} \ell_{\sigma}(I_j)\,.
\eeq

\vspace{0.2cm}

\noindent $\bullet$ \textbf{Embedded pairs of pants.} Each $x \in \beta_1'$ canonically determines an embedded pair of pants $P_{x} \subset \Sigma_{g,n}$ which bounds $\beta_1$, as follows. Consider a thickening of $\gamma_{x}$. This determines two free homotopy classes (relative to the boundary) of closed curves -- see the picture below. We denote $\alpha$ and $\alpha'$ their shortest geodesic representative. Given the geometry, $\beta_1$ together with $\alpha$ and $\alpha'$ is bounding an embedded pair of pants, which is our $P_{x}$.

\begin{figure}[h!]
\begin{center}
\includegraphics[width=\textwidth]{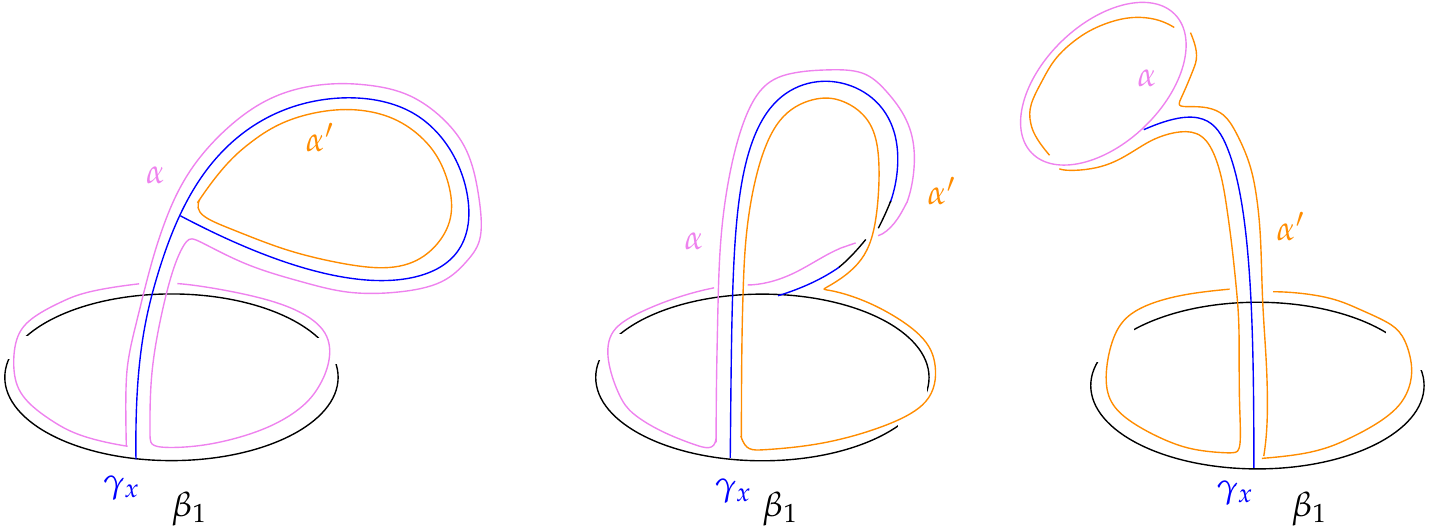}
\end{center}
\end{figure}

Conversely, if $Q$ is a embedded pair of pants bounding $\beta_1$, we can identify several orthogeodesics issueing from a point in $\beta_1$, which are canonically associated to $P$. We again denote $\alpha$ and $\alpha'$ the two other boundaries of $Q$. These orthogeodesics are
\begin{itemize}
\item[$(i)$] the shortest geodesics joining $\beta_1$ to $\alpha$, and $\beta_1$ to $\alpha'$. We denote $r$ and $r'$ their origin in $\beta_1$. These geodesics cut $Q$ into two hyperbolic right-angled hexagons $\{H_1,H_2\}$, which are exchanged by an isometric involution.
\item[$(ii)$] the geodesic issuing from a point $p_1 \in \beta_1 \cap H_1$ orthogonally to $\beta_1$, and coming back orthogonally to $\beta_1$ at a point $p_2 \in \beta_1 \cap H_2$. This is the seam of the pair of pants.
\item[$(iii)$] the geodesic issued from a point in $o_1 \in \beta_1 \cap H_1$ which spirals around $\alpha$, and the symmetric one issued from a point $o_2 \in \beta_1 \cap H_2$.
\item[$(iii)'$] the geodesic issued from a point in $p_1' \in \beta_1 \cap H_1$ which spirals around $\alpha'$, and the symmetric one issued from a point $o_2' \in \beta_1 \cap H_2$.
\end{itemize}
Let us denote $\ell_1 = \ell_{\sigma}(\beta_1)$, $\ell_2 = \ell_{\sigma}(\alpha)$ and $\ell_3 = \ell_{\sigma}(\alpha')$, and define $\tfrac{1}{2} d(\ell_1,\ell_2,\ell_3)$ to be the length of the segment $(o_1',r')$ in $\beta_1 \cap H_1'$. We now discuss two cases. Exchanging the role of $\alpha$ and $\alpha'$, we see that the length of the segment $(o_1,r)$ in $\beta_1 \cup H_1$ must be $\tfrac{1}{2}d(\ell_1,\ell_3,\ell_2)$. We now distinguish two cases.

\begin{figure}[h!]
\begin{center}
\includegraphics[width=0.9\textwidth]{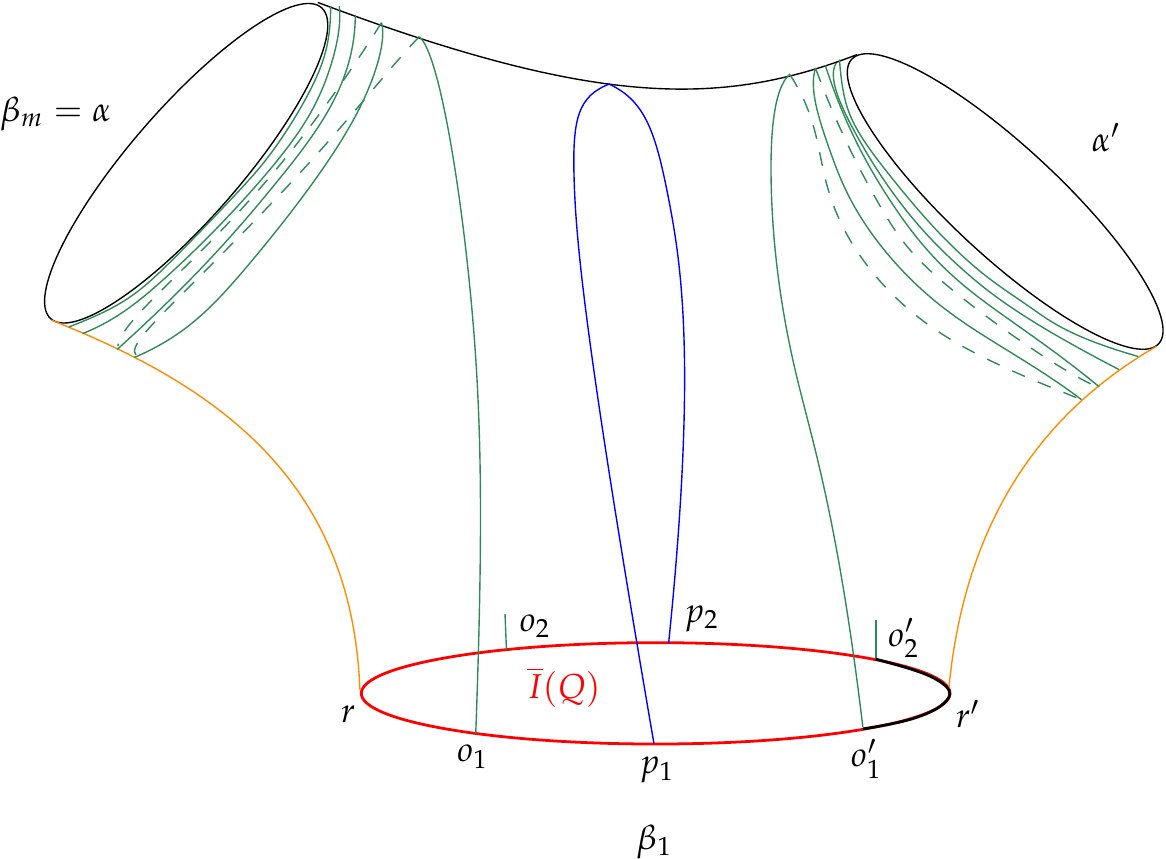}
\end{center}
\end{figure}
\begin{itemize}
\item[$(\textbf{B})$] On top of $\beta_1$, $Q$ bounds another boundary of $\Sigma_{g,n}$. We can say this is $\alpha = \beta_m$ for some $m \in \{2,\ldots,n\}$. According to Lemma~\ref{clasgeod}, $o_2'$ and $o_1'$ do not belong to $\beta_1'$. Let $\overline{I}(Q) \subset \beta_1$ be the open segment $(o_2',o_1')$ which contains the points $p_2$, $o_1$, $r$, $o_2$ and $p_1$. Looking at the behavior of $\gamma_{x}$ for $x \in \overline{I}(Q)$, we see that $\overline{I}(Q) \subset \beta_1'$ -- for instance, the five points we mentioned in the interior of $\overline{I}(Q)$ are isolated points in $E$ according to the classification of Lemma~\ref{clasgeod}, so belong to $\beta_1'$, and for any $x \in \overline{I}(Q)$, $P_{x} = Q$. 
\begin{figure}[h!]
\begin{center}
\includegraphics[width=0.9\textwidth]{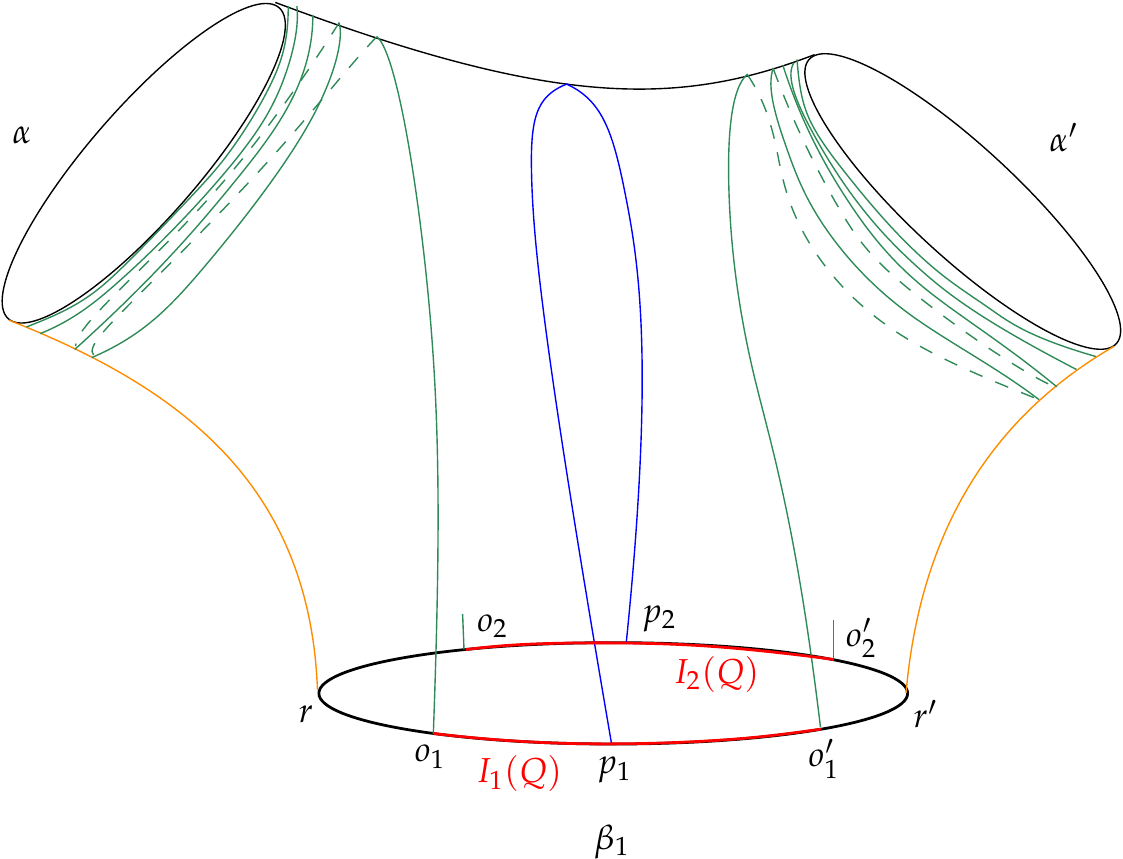}
\end{center}
\end{figure}
\item[$(\textbf{C})$] $\alpha$ and $\alpha'$ are interior to $\Sigma_{g,n}$. Let $I_1(Q) \subset \beta_1$ be the segment from $o_1$ to $o_1'$ containing $p_1$, $I_2(Q) \subset \beta$ the segment from $o_2$ to $o_2'$ containing $p_2$, and $\overline{I}(Q) = I_1(Q) \dot{\cup} I_2(Q)$. Looking at the behavior of $\gamma_{x}$ for $x \in \overline{I}(Q)$, we find again that for any $x \in \overline{I}(Q)$, $P_{x} = Q$.
\end{itemize}
We conclude that $P_{x}$ for $x \in \beta_1'$ only depends on the open segment $I_j$ to which $x$ belongs, that any embedded pair of pants $Q$ is realized as a $P_{x}$ for some $x \in \beta_{1}'$, and we have identified in the connected components of $\overline{I}(Q)$ the segments in \eqref{tegh} which are associated in the previous construction with the embedded pair of pants $Q$. We also have
\begin{itemize}
\item[$(\textbf{B})$] $\ell_{\sigma}(\overline{I}(Q)) = L_1 - d(L_1,L_m,\ell_{\sigma}(\alpha'))$.
\item[$(\textbf{C})$] $\ell_{\sigma}(\overline{I}(Q)) = \Big(\tfrac{L_1}{2} - d(L_1,\ell_{\sigma}(\alpha),\ell_{\sigma}(\alpha'))\Big) + \Big(\tfrac{L_1}{2} - d(L_1,\ell_{\sigma}(\alpha'),\ell_{\sigma}(\alpha))\Big)$.
\end{itemize}
and the sum \eqref{tegh} becomes
\beq 
\label{tegh2} L_1 = \sum_{Q \in \mathcal{Q}} \ell_{\sigma}(\overline{I}(Q))\,,
\eeq
where $\mathcal{Q}$ is the set of embedded pair of pants  in $\Sigma_{g,n}$ bounding $\beta_1$ and having geodesic boundaries.

\vspace{0.2cm}

\noindent $\bullet$ \textbf{Mapping class group orbits.} $\mathcal{Q}$ is a countable set on which the mapping class group $\Gamma_{g,n}$ acts. The equivalence classes are in correspondence with the various topologies that the surface $\Sigma_{g,n} \setminus Q$ can have.
\begin{itemize}
\item[$(\textbf{I})$] for $m \in \{2,\ldots,n\}$, the orbit $\mathcal{Q}_{g,n - 1}(m)$ is the set of $Q$s bounding $\beta_1$ and $\beta_m$.
\item[$(\textbf{I}')$] $\mathcal{Q}_{g-1,n + 1}$ is the set of $Q$s bounding $\beta_1$ and two curves $\alpha$ and $\alpha'$ in the interior of $\Sigma_{g,n}$, such that $\Sigma_{g,n} \setminus Q$ remains connected -- and thus has genus $g - 1$ and $n + 1$ boundaries.
\item[$(\textbf{II})$] for $J \dot{\cup} J'$ a partition of $\{\beta_2,\ldots,\beta_m\}$ and $h + h' = g$, $\mathcal{Q}_{h,J;h',J'}$ is the set of $Q$s bounding $\beta_1$ and a non ordered pair of curves $\alpha$ and $\alpha'$ in the interior of $\Sigma_{g,n}$, such that $\Sigma_{g,n}\setminus Q$ is not connected -- it then consists of two connected components, and the one bounding $\alpha$ has other boundaries $(\beta_j)_{j \in J}$ and genus $h$, while the one bounding $\alpha'$ has other boundaries $(\beta_j)_{j \in J'}$ and genus $h'$.
\end{itemize}
In particular, there are finitely many equivalence classes, which are precisely in correspondence with the terms in the \textsc{tr} equation \eqref{TReq2}.

\vspace{0.2cm}

\noindent $\bullet$ \textbf{Summary: generalized McShane identity.} If we collect the terms in \eqref{tegh2} by mapping class group orbits, we arrive to the generalized McShane identity proved as an intermediate result by Mirzakhani.
\bea
&& L_1 \nonumber \\
& = & \sum_{m = 2}^n \sum_{\substack{Q \in \mathcal{Q}_{g,n-1}(m) \\ \partial Q = \beta_1 \cup \beta_m \cup \alpha'}} \big(L_1 - d(L_1,L_m,\ell_{\sigma}(\alpha')\big) \nonumber \\
& +& \bigg[  \sum_{\substack{Q \in \mathcal{Q}_{g-1,n+1} \\ \partial Q = \beta_1 \cup \alpha \cup \alpha'}} \!\!+\!\!\sum_{\substack{J \dot{\cup} J' = \{\beta_2,\ldots,\beta_n\} \\ h + h' = g}} \sum_{\substack{Q \in \mathcal{Q}_{h,J;h',J'} \\ \partial Q = \beta_1 \cup \alpha \cup \alpha'}}\bigg] \Big(L_1 - d(L_1,\ell_{\sigma}(\alpha),\ell_{\sigma}(\alpha')) - d(L_1,\ell_{\sigma}(\alpha'),\ell_{\sigma}(\alpha))\Big)\,, \nonumber \\
\label{mccc} &&
\eea
We stress that $((J,h),(J',h'))$ is an ordered pair in the last sum. This identity expresses the constant function $L_1$ on $\mathcal{T}_{g,n}(L_1,\ldots,L_n)$, as a sum, over embedded pairs of pants bounding $\beta_1$, of non-trivial functions over $\mathcal{T}_{g,n}(L_1,\ldots,L_n)$. Note that the right-hand side is manifestly $\Gamma_{g.n}$-invariant as we are summing over all mapping class group orbits.

The original McShane identity concerned the once-punctured torus (the case $g = 1$, $n = 1$ and $L_1 = 0$), in which the pair of pants is self-glued and the sum over $\mathcal{Q}$ is a sum over simple closed geodesics \cite{McShane}. We excluded on purpose this case in our exposition by assuming $2g - 2 + n \geq 2$. In fact, from the original McShane identity one can rederive that $V_{1,1}(L = 0) = \zeta(2)$.

\vspace{0.2cm}

\noindent $\bullet$ \textbf{Integration over the moduli space.} Let us divide \eqref{mccc} by $L_1$, and integrate over the moduli space $\mathcal{M}_{g,n}(L_1,\ldots,L_n)$. The left-hand side is by definition $V_{g,n}(L_1,\ldots,L_n)$. In the right-hand side, let us consider a $\Gamma_{g,n}$-orbit $\mathcal{O}$ of $\mathcal{Q}$, and fix $Q_0 \in \mathcal{O}$. Let ${\rm Stab}(Q_0) \subset \Gamma_{g,n}$ be the stabilizer of $Q_0$. We have
$$
\sum_{Q \in \mathcal{O}} \ell_{\sigma}(\overline{I}(Q)) = \sum_{\varphi \in \Gamma_{g,n}/{\rm Stab}(Q_0)} \ell_{\varphi^*\sigma}(\overline{I}(Q_0))\,,
$$
and therefore
$$
\int_{\mathcal{M}_{g,n}(L_1,\ldots,L_n)} \bigg(\sum_{Q \in \mathcal{O}} \ell_{\sigma}(\overline{I}(Q))\bigg) \nu_{{\rm WP}} \ = \int_{\mathcal{T}_{g,n}(L_1,\ldots,L_n)/{\rm Stab}(Q_0)} \ell_{\sigma}(\overline{I}(Q_0))\,\nu_{{\rm WP}}\,.
$$
To compute this integral, we should describe a fundamental domain $\mathcal{D}(Q_0)$ of $\mathcal{T}_{g,n}(L_1,\ldots,L_n)/{\rm Stab}(Q_0)$. As $\nu_{{\rm WP}}$ has a simple expression \eqref{WPform}-\eqref{WPform2}, it is convenient to use Fenchel-Nielsen coordinates for a pair of pants decomposition of $\Sigma_{g,n}$ containing $Q_0$. Let $S$ be the subset of $\{\alpha,\alpha'\}$ consisting of the curves interior to $\Sigma_{g,n}$, and $\ell(\gamma),\theta(\gamma)$ be the length and twist coordinates attached to a curve $\gamma \in S$. The remaining set of coordinates parametrizes the Teichm\"uller space of the bordered surface $\Sigma_{g,n}\setminus Q_0$ in which the curves $\gamma \in S$ have fixed length $\ell(\gamma)$. As $\mathcal{D}(Q_0)$ we can take a domain which is cut out by the equation $0 \leq \theta(\gamma) < \ell(\gamma)$ for each $\gamma \in S$, together with additional constraints depending on the orbit $\mathcal{O}$. In particular, one has to carefully take into account the fact that $\alpha$ and $\alpha'$ are not \textit{a priori} ordered in the case (\textbf{C}). The net effect is a factor $\tfrac{1}{2}$ when $\ell(\alpha),\ell(\alpha')$ are integrated over $\mathbb{R}_{> 0}^2$. We also remark that the integral over the twist coordinates $\theta(\gamma)$ for $\gamma \in S$ always yields a factor of $\ell(\gamma)$. The outcome is
\bea
&& V_{g,n}(L_1,\ldots,L_n) \nonumber \\
& = & \sum_{m = 2}^n \int_{\mathbb{R}_{> 0}} \mathcal{B}(L_1,L_m,\ell)\,V_{g,n - 1}(\ell,L_2,\ldots,\widehat{L_m},\ldots,L_n) \nonumber \\
&& + \tfrac{1}{2} \int_{\mathbb{R}_{> 0}^2} \dd \ell \dd \ell'\,\mathcal{C}(L_1,\ell,\ell')\bigg(V_{g - 1,n + 1}(\ell,\ell',L_2,\ldots,L_n) \nonumber \\
&& \phantom{\tfrac{1}{2} \int_{\mathbb{R}_{> 0}^2} \dd \ell \dd \ell'\,\mathcal{C}(L_1,\ell,\ell')}\quad  + \sum_{\substack{J_1 \dot{\cup} J_2 = \{L_2,\ldots,L_n\} \\ h_1 + h_2 = g}} V_{g_1,1 + |J_1|}(\ell,J_1)V_{g_2,1 + |J_2|}(\ell',J_2)\bigg)\,, \nonumber \\
\label{mccc2} && 
\eea
where
\bea
\mathcal{B}(L_1,L_2,\ell) & = & \frac{\ell}{L_1}\big(L_1 - d(L_1,L_2,\ell)\big)\,, \nonumber \\
\mathcal{C}(L_1,\ell,\ell') & = & \frac{2\ell\ell'}{L_1}\big(L_1 - d(L_1,\ell,\ell') - d(L_1,\ell',\ell)\big)\,. \nonumber
\eea

\vspace{0.2cm}

\noindent $\bullet$ \textbf{Hyperbolic trigonometry.} It remains to compute the function $d(\ell_1,\ell_2,\ell_3)$ which appears in the generalized McShane identity \eqref{mccc}, and in the recursive formula \eqref{mccc2} for the volumes. Buser describes in \cite{Buser} a systematic way to derive all useful identities in hyperbolic trigonometry.

We first apply these formulas to the right-angled hyperbolic hexagon $H_1$, which has lengths $\tfrac{L_1}{2}$, $\tfrac{\ell_2}{2}$, $\tfrac{L_3}{2}$, $\tfrac{\ell_1}{2}$, $\tfrac{L_2}{2}$, $\tfrac{\ell_3}{2}$.
\begin{figure}[h!]
\begin{center}
\includegraphics[width=0.45\textwidth]{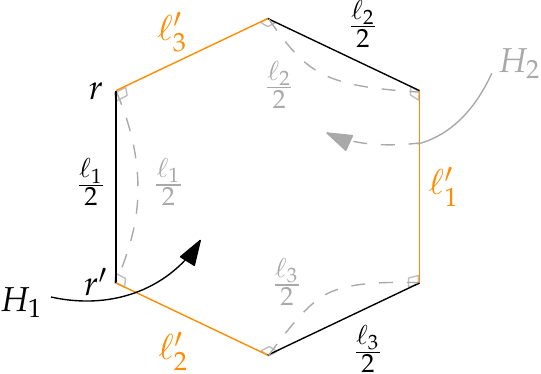}
\caption{\label{Fig5} Right-angled hexagon -- its isometric copy with which it forms a hyperbolic pair of pants is indicated with dashed edges.}
\end{center}
\end{figure}

\vspace{-0.2cm}

Any three of the lengths completely determine the three others. The formula we need is
\beq
\label{Bus1} {\rm cosh}(\ell_2') = \frac{{\rm cosh}(\frac{\ell_2}{2}) + {\rm cosh}(\frac{\ell_1}{2}){\rm cosh}(\frac{\ell_3}{2})}{{\rm sinh}(\frac{\ell_1}{2}){\rm cosh}(\frac{\ell_3}{2})}\,,
\eeq
and by exchange the role of $2$ and $3$ we obtain the formula for $\ell_3'$ in terms of $(\ell_1,\ell_2,\ell_3)$.

Then, consider the universal cover $\tilde{Q}$ of $Q$ based at $o_1'$, and equip it with the pullback of the hyperbolic metric on $Q$. This is again a hyperbolic metric. The lift to $\tilde{Q}$ of $\alpha'$, and of the orthogeodesic issued from $p_1'$ which spirals around $\alpha'$ -- here considered as the third boundary after $\beta_1$ and $\alpha$ -- have infinite hyperbolic length. We denote $\tilde{\alpha}'$ and $\tilde{\gamma}$ these lifts. They approach each other with a vanishing angle. Then, the lift of the segment $[o_1',r'] \subset \beta_1 \cap H_1$ to $\tilde{Q}$ has length $\tfrac{1}{2}d(\ell_1,\ell_2,\ell_3)$, and the lift of the shortest geodesic from $\beta_1$ to $\alpha'$ has length $\ell_2'$. By construction, these two lifted segments meet at $r'$ forming a right-angle. We in fact have a quadrilateral with three right angles and a vanishing angle -- see the picture below.

\begin{figure}[h!]
\begin{center}
\includegraphics[width=0.37\textwidth]{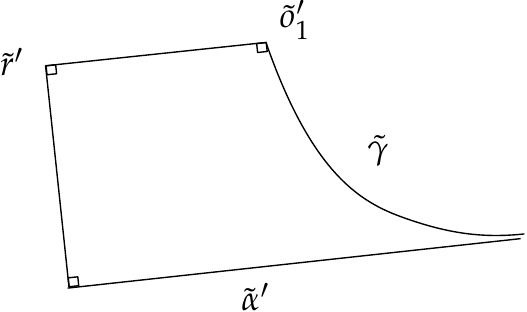}
\end{center}
\end{figure}

\vspace{-0.2cm}
\noindent The formulas of \cite{Buser} yield ${\rm sinh}(\ell_2')\,{\rm sinh}\big(\tfrac{d}{2}\big) = 1$. If we eliminate $\ell_2'$ with \eqref{Bus1}, we obtain
$$
d(\ell_1,\ell_2,\ell_3) = \ln\bigg(\frac{{\rm cosh}(\frac{\ell_2}{2}) + {\rm cosh}(\frac{\ell_1 + \ell_3}{2})}{{\rm cosh}(\frac{\ell_2}{2}) + {\rm cosh}(\frac{\ell_1 - \ell_3}{2})}\bigg)\,.
$$
Equivalently
\bea
&& d(\ell_1,\ell_2,\ell_3) \nonumber \\ 
& = & \ln\bigg(\frac{e^{\frac{\ell_1 + \ell_3}{2}}}{e^{\frac{\ell_3 - \ell_1}{2}}}\cdot \frac{1 + e^{-\ell_3 - \ell_1} + e^{\frac{-\ell_3 + \ell_2 - \ell_1}{2}} + e^{-\frac{\ell_3 + \ell_2 + \ell_1}{2}}}{1 + e^{-\ell_3 + \ell_1} + e^{\frac{-\ell_3 + \ell_2 + \ell_1}{2}} + e^{-\frac{-\ell_3 - \ell_2 + \ell_1}{2}}}\bigg) \nonumber \\ 
& = & \ell_1 - \tfrac{1}{2}\big(f(\ell_3 + \ell_2 + \ell_1) - f(\ell_3 + \ell_2 - \ell_1) \nonumber \\
&& \qquad + f(\ell_3 - \ell_2 + \ell_1) - f(\ell_3 - \ell_2 - \ell_1)\big) \nonumber
\eea
in terms of the function $f(z) = -2\ln(1 + e^{-z/2})$. Therefore, the functions appearing in \eqref{mccc} are
\bea
\ell_1 - d(\ell_1,\ell_2,\ell_3) & = & \tfrac{1}{2}\big(f(\ell_3 + \ell_2 + \ell_1) - f(\ell_3 + \ell_2 - \ell_1) \nonumber \\
&& \quad\quad + f(\ell_3 - \ell_2 + \ell_1) - f(\ell_3 - \ell_2 - \ell_1)\big)\,, \nonumber \\
\ell_1 - d(\ell_1,\ell_2,\ell_3) - d(\ell_1,\ell_3,\ell_2) & = & f(\ell_3 + \ell_2 + \ell_1) - f(\ell_3 + \ell_2 - \ell_1)\,. \nonumber
\eea
where we have used $f(z) - f(-z) = z$ in the last line. This also leads to the claimed expressions \eqref{Ccal} for $\mathcal{B}$ and $\mathcal{C}$.

\newpage

\providecommand{\bysame}{\leavevmode\hbox to3em{\hrulefill}\thinspace}
\providecommand{\MR}{\relax\ifhmode\unskip\space\fi MR }
\providecommand{\MRhref}[2]{%
  \href{http://www.ams.org/mathscinet-getitem?mr=#1}{#2}
}
\providecommand{\href}[2]{#2}


\begin{thebibliography}{10}

\bibitem{Abrams}
L.~Abrams, \emph{Two-dimensional topological quantum field theories and
  {F}robenius algebras}, J. Knot Theory and Applications \textbf{5} (1996),
  no.~5, 569--587.

\bibitem{TRABCD}
J.E. Andersen, G.~Borot, L.O. Chekhov, and N.~Orantin, \emph{The {ABCD} of
  topological recursion},  (2017), math-ph/1703.03307.

\bibitem{ABO1}
J.E. Andersen, G.~Borot, and N.~Orantin, \emph{Modular functors, cohomological
  field theories and topological recursion},  (2015), math-ph/1509.01387.

\bibitem{At2}
M.F. Atiyah, \emph{Topological quantum field theory}, Publications
  math\'ematiques de l'IHES \textbf{68} (1988), 175--186.

\bibitem{BB}
B.~Bakalov and A.Jr. Kirillov, \emph{Lectures on tensor categories and modular
  functors}, University Lecture Series, vol.~21, American Mathematical Society,
  Providence, RI, 2001.

\bibitem{FBclass}
K.~Behrend and B.~Fantechi, \emph{The intrinsic normal cone}, Invent. Math.
  \textbf{128}, no.~1, 45--88.

\bibitem{BirmanSeries}
J.S. Birman and C.~Series, \emph{Geodesics with bounded intersection number on
  surfaces are sparsely distributed}, Topology \textbf{24} (1985), 217--225.

\bibitem{Buser}
P.~Buser, \emph{Geometry and spectral of compact {R}iemann surfaces}, Progress
  in Mathematics, vol. 106, BIrkh\"auser, 1992.

\bibitem{Dubrovin}
B.~Dubrovin, \emph{Geometry of {$2d$} topological field theories}, Lecture
  Notes in Mathematics, vol. 1620, pp.~120--348, Springer, Berlin, 1996,
  hep-th/9407018.

\bibitem{DBOSS}
P.~Dunin-Barkowski, N.~Orantin, S.~Shadrin, and L.~Spitz, \emph{Identification
  of the {G}ivental formula with the spectral curve topological recursion
  procedure}, Commun. Math. Phys. \textbf{328} (2014), no.~2, 669--700,
  math-ph/1211.4021.

\bibitem{EHX}
T.~Eguchi, K.~Hori, and C.-S. Xiong, \emph{Quantum cohomology and {V}irasoro
  algebra}, Phys. Lett. B \textbf{402} (1994), 71--80, hep-th/9703086.

\bibitem{EOFg}
B.~Eynard and N.~Orantin, \emph{Invariants of algebraic curves and topological
  expansion}, Commun. Number Theory and Physics \textbf{1} (2007), no.~2,
  math-ph/0702045.

\bibitem{EOwp}
\bysame, \emph{Weil-{P}etersson volume of moduli spaces, {M}irzakhani's
  recursion and matrix models},  (2007), math-ph/0705.3600.

\bibitem{FSZ}
C.~Faber, S.~Shadrin, and D.~Zvonkine, \emph{Tautological relations and the
  {$r$}-spin {W}itten conjecture}, Annales scientifiques de l'{\'{E}}cole
  {N}ormale {S}up{\'{e}}rieure \textbf{43} (2010), no.~4, 621--665,
  math.AG/0612510.

\bibitem{GetzlerVir}
E.~Getzler, \emph{The {V}irasoro conjecture for {G}romov-{W}itten invariants},
  Algebraic geometry: Hirzebruch 70 (Providence R.I.), Contemp. Math., vol.
  241, American Mathematical Society, 1998, pp.~147--176.

\bibitem{Givquant}
A.B. Givental, \emph{{G}romov--{W}itten invariants and quantization of
  quadratic hamiltonians}, Moscow J. Math. \textbf{1} (2001), no.~4, 551--568,
  math.AG/0108100.

\bibitem{Giventals}
\bysame, \emph{Semisimple {F}robenius structures at higher genus}, Int. Math.
  Res. Not. (2001), 1265--1286, math.AG/0008067.

\bibitem{Kontsevich}
M.~Kontsevich, \emph{Intersection theory on the moduli space of curves and the
  matrix {A}iry function}, Commun. Math. Phys. \textbf{147} (1992), 1--23.

\bibitem{KMCohFT}
M.~Kontsevich and Yu. Manin, \emph{Gromov-{W}itten classes, quantum cohomology
  and enumerative geometry}, Commun. Math. Phys. \textbf{164} (1994), no.~3,
  525--562.

\bibitem{KSTR}
M.~Kontsevich and Y.~Soibelman, \emph{Airy structures and symplectic geometry
  of topological recursion}, math.AG/1701.09137.

\bibitem{Zvonkin}
A.~Marian, D.~Oprea, R.~Pandharipande, A.~Pixton, and D.~Zvonkine, \emph{The
  {C}hern character of the {V}erlinde bundle over the moduli space of curves},
  J. reine angew. Math. (2015), math.AG/1311.3028.

\bibitem{McShane}
G.~McShane, \emph{Simple geodesics and a series constant over {T}eichm{\"u}ller
  space}, Invent. Math. \textbf{132} (1998), 607--632.

\bibitem{Mirza1}
M.~Mirzakhani, \emph{Simple geodesics and {W}eil-{P}etersson volumes of moduli
  spaces of bordered {R}iemann surfaces}, Invent. math. \textbf{167} (2007),
  179--222.

\bibitem{Penner}
R.C. Penner, \emph{{W}eil-{P}etersson volumes}, J. Diff. Geom. \textbf{35}
  (1992), no.~3, 559--608.

\bibitem{Teleman}
C.~Teleman, \emph{The structure of {$2d$} semi-simple field theories}, Invent.
  Math. \textbf{188} (2012), 525--588, math.AT/0712.0160.

\bibitem{TUY}
A.~Tsuchiya, K.~Ueno, and Y.~Yamada, \emph{Conformal field theory on universal
  family of stable curves with gauge symmetries}, Advanced Studies in Pure
  Mathematics \textbf{19} (1989), 459--566.

\bibitem{Witten}
E.~Witten, \emph{Two dimensional gravity and intersection theory on moduli
  space}, Surveys in Diff. Geom. \textbf{1} (1991), 243--310.

\bibitem{Wolpert1}
S.A. Wolpert, \emph{On the {W}eil-{P}etersson geometry of the moduli space of
  curves}, Amer. J. Math. \textbf{107} (1985), no.~4, 969--997.

\bibitem{Wolpertlecture}
\bysame, \emph{{M}irzakhani's volume recursion and approach for the
  {W}itten-{K}ontsevich theorem on moduli tautological intersection numbers},
  IAS/Park City Mathematics Series, vol.~20, AMS, 2013, math.DG/1108.0174.

\end{thebibliography}
\end{document}